\documentclass[10pt,a4paper]{article}

\usepackage{amsmath,amsgen,latexsym}
\usepackage{amstext,amssymb,amsfonts,latexsym}
\usepackage{theorem}
\usepackage{pifont}
\usepackage{graphicx}
\usepackage{relsize}

 \newcommand{\bs}{\bigskip}
 \newcommand{\ms}{\medskip}
 \newcommand{\n}{\noindent}
 \newcommand{\s}{\smallskip}
 \newcommand{\hs}[1]{\hspace*{ #1 mm}}
 \newcommand{\vs}[1]{\vspace*{ #1 mm}}

 \newcommand{\setempty}{\varnothing}
 
 \newcommand{\nat}{\mathbb{N}}
 
 \newcommand{\integer}{\mathbb{Z}}

 \newcommand{\co}{\mathrm{co}\mbox{-}}

 \newcommand{\BB}{{\cal B}}
 
 \newcommand{\CC}{{\cal C}}
 \newcommand{\FF}{{\cal F}}

 \newcommand{\HH}{{\cal H}}
 
 \newcommand{\KK}{{\cal K}}
 \newcommand{\LL}{{\cal L}}
 \newcommand{\NN}{{\cal N}}
 \newcommand{\GG}{{\cal G}}
 \newcommand{\MM}{{\cal M}}

 \newcommand{\PP}{{\cal P}}

 \newcommand{\UU}{{\cal U}}
 \newcommand{\VV}{{\cal V}}

 \newcommand{\dl}{\mathrm{L}}
 \newcommand{\nl}{\mathrm{NL}}

 \newcommand{\pp}{\mathrm{PP}}

 \newcommand{\poly}{\mathrm{poly}}

 \newcommand{\fp}{\mathrm{FP}}
 \newcommand{\sharpp}{\#\mathrm{P}}
 \newcommand{\gapp}{\mathrm{GapP}}

 \newcommand{\reg}{\mathrm{REG}}
 \newcommand{\cfl}{\mathrm{CFL}}

\theoremstyle{plain}
\theoremheaderfont{\bfseries}
 \newtheorem{theorem}{Theorem}[section]
 \newtheorem{lemma}[theorem]{Lemma}
 \newtheorem{proposition}[theorem]{{\bf Proposition}}
 \newtheorem{corollary}[theorem]{Corollary}

 {\theorembodyfont{\rmfamily}  \newtheorem{definition}[theorem]{Definition}}
 {\theorembodyfont{\rmfamily} \newtheorem{example}[theorem]{Example}}
 {\theorembodyfont{\rmfamily} }

 \newenvironment{proofsketch}{\par \noindent
            {\bf Proof Sketch. \hs{2}}}{\hfill$\Box$ \vspace*{3mm}}

 \newenvironment{proofof}[1]{\vspace*{5mm} \par \noindent
         {\bf Proof of #1.\hs{2}}}{\hfill$\Box$ \vspace*{3mm}}

 \newenvironment{proofsketchof}[1]{\vspace*{5mm} \par \noindent
         {\bf Proof Sketch of #1.\hs{2}}}{\hfill$\Box$ \vspace*{3mm}}

 \newenvironment{proof}{\par \noindent
            {\bf Proof. \hs{2}}}{\hfill$\Box$ \vspace*{3mm}}

 \newcommand{\ceilings}[1]{\lceil #1 \rceil}
 \newcommand{\floors}[1]{\lfloor #1 \rfloor}
 \newcommand{\pair}[1]{\langle #1 \rangle}

 \newcommand{\parity}{\oplus}

\newcommand{\ignore}[1]{}

\newcommand{\dollar}{\$}

 \newcommand{\oned}{1\mathrm{D}}

 \newcommand{\onep}{1\mathrm{P}}

 \newcommand{\onen}{1\mathrm{N}}
 \newcommand{\twod}{2\mathrm{D}}
 
 \newcommand{\twon}{2\mathrm{N}}

 \newcommand{\onecequal}{1\mathrm{C_{=}}}
 
 \newcommand{\oneparity}{1\oplus}

 \newcommand{\twoparity}{2\oplus}
 
 \newcommand{\twocequal}{2\mathrm{C_{=}}}
 \newcommand{\oneu}{1\mathrm{U}}

\newcommand{\sharpl}{\mathrm{\#\mathrm{L}}}

\newcommand{\gapl}{\mathrm{GapL}}

\newcommand{\cequall}{\mathrm{C}_{=}\mathrm{L}}

\newcommand{\onesharp}{1\#}
\newcommand{\onegap}{1\mathrm{Gap}}
\newcommand{\onef}{1\mathrm{F}}

\newcommand{\parityl}{\oplus\mathrm{L}}
\newcommand{\onesp}{1\mathrm{SP}}
\newcommand{\onedpd}{1\mathrm{DPD}}
\newcommand{\onenpd}{1\mathrm{NPD}}

\newcommand{\pl}{\mathrm{PL}}
\newcommand{\twosharp}{2\#}
\newcommand{\twogap}{2\mathrm{Gap}}

\newcommand{\dbraleft}{\:[\!\!\![\;}
\newcommand{\dbraright}{\;]\!\!\!]\,}

\begin{document}

\pagestyle{plain}
\pagenumbering{arabic}
\setcounter{page}{1}
\setcounter{footnote}{0}

\begin{center}
{\Large {\bf Power of Counting by Nonuniform Families of \s\\
Polynomial-Size
Finite Automata}}\footnote{This corrects and extends a preliminary report published in the Proceedings of the 24th International Symposium on Fundamentals of Computation Theory (FCT 2023), Trier, Germany, September 18--24, 2023, Lecture Notes in Computer Science, vol. 14292, pp. 421--435, Springer Cham, 2023.} \bs\\
{\sc Tomoyuki Yamakami}\footnote{Present Affiliation: Faculty of
Engineering, University of Fukui, 3-9-1 Bunkyo, Fukui 910-8507,
Japan}
\bs\\
\end{center}


\begin{abstract}
Lately, there have been intensive studies on strengths and limitations of nonuniform families of promise decision problems solvable by various types of polynomial-size finite automata families, where ``polynomial-size'' refers to the polynomially-bounded state complexity of a finite automata family.
In this line of study, we further expand the scope of these studies to families of partial counting and gap functions, defined in terms of nonuniform families of polynomial-size nondeterministic finite automata, and their relevant families of promise decision problems. Counting functions have an ability of counting the number of accepting computation paths produced by nondeterministic finite automata.
With no unproven hardness assumption, we show numerous separations and collapses of complexity classes of those partial counting and gap function families and their induced promise decision problem families.
We also investigate their relationships to pushdown automata families of polynomial stack-state complexity.

\ms

\n{\bf Key words.} nonuniform polynomial state complexity, counting functions, gap functions, counting complexity classes, closure properties, stack-state complexity
\end{abstract}


\sloppy
\section{Background and New Approach}\label{sec:introduction}

We quickly review basic background knowledge on nonuniform (polynomial) state complexity and counting and gap functions and discuss a new approach taken in this work toward the better understandings of counting.

\subsection{Nonuniform Automata Families and Counting}

In computational complexity theory, ``nonuniformity'' has played a distinguishing role.
Nonuniformity is often formulated by the use of nonuniform families of underlying machines. Within the framework of finite automata, mostly \emph{state complexity of transformation} has been studied in the literature (see surveys, e.g., \cite{GMR+16,HK11}).
Its quantum version was also investigated in  \cite{VY15}. This notion measures the increase of inner states by converting one type of automata of $n$ states to another type of automata.

From slightly different perspective, Berman and Lingas \cite{BL77} and Sakoda and Sipser \cite{SS78} studied in the late 1970s the computational power of  ``nonuniform'' families $\{M_n\}_{n\in\nat}$ of finite automata $M_n$, indexed by natural numbers $n$, built with polynomially many inner states (or, of \emph{polynomial state complexity}) in $n$.
A series of subsequent studies \cite{Gef12,Kap09, Kap12,Kap14, KP15, Yam19a, Yam19b,Yam22a} along the line of Sakoda and Sipser has made a significant contribution to a development of the \emph{theory of nonuniform (polynomial) state complexity}.
A further expansion of the scope of this theory has been expected to promote our basic understandings of the theory.

To simplify our further discussion, we here introduce the basic notations $\oned$ and $\onen$ of Sakoda and Sipser to denote respectively the collections of
nonuniform families of promise decision problems\footnote{A \emph{promise decision problem} over alphabet $\Sigma$ is a pair $(A,R)$ satisfying $A,R\subseteq \Sigma^*$ and $A\cap R=\setempty$. In particular, when $A\cup R=\Sigma^*$, $(A,R)$ is identified with the language $A$.} over fixed
alphabets\footnote{In this work, for any given family, we always fix an input alphabet for all promise problems in this family. This situation slightly differs from \cite{Kap09,Kap12,Kap14,KP15}, in which alphabets may vary  according to promise problems in the family.}
solvable by one-way deterministic finite automata (or 1dfa's, for
short) and by one-way nondeterministic finite automata (or 1nfa's) using  polynomially many inner states.

After an early study of \cite{BL77,SS78}, Kapoutsis \cite{Kap09,Kap12}  revitalized the study of nonuniform families of finite automata by expanding the scope of underlying machines to
probabilistic and alternating finite automata.
Geffert \cite{Gef12} further investigated the behaviors of nonuniform families of alternating finite automata. Yamakami expanded underlying machines to quantum finite automata \cite{Yam22a}, nondeterministic finite automata with fewer accepting computation paths \cite{Yam22b}, and width-bounded two-way nondeterministic finite automata (or 2nfa's) \cite{Yam19a} in direct connection to the \emph{linear space hypothesis} (LSH) \cite{Yam23a}. The  ``relativization'' of nonuniform (polynomial) sate complexity classes was introduced in  \cite{Yam19b}.
As another natural extension of finite automata families, nonuniform families of pushdown automata were studied in \cite{Yam21a}, where polynomial state complexity is replaced by \emph{polynomial stack-state complexity} (see Section \ref{sec:various-type} for its definition).

Nondeterminism has been a core concept of theoretical computer science against the opposite concept of determinism. To fully understand nondeterminism, it is important to investigate the nature of nondeterminism, in particular, from the viewpoint of comparing between the number of accepting computation paths and that of rejecting computation paths of nondeterministic computation.
Unambiguous computation, for instance, meets the criteria of having at most one accepting computation path.
Even acceptance probability of probabilistic computation is often understood as a ratio of the number of accepting computation paths of nondeterministic computation over the total number of halting computation paths.
Unambiguous computation and unbounded-error probabilistic computation of finite automata families were already studied in \cite{Kap09,Kap12,Yam22a} in terms of the nonuniform polynomial state complexity classes $\oneu$ and $\onep$.
These studies have signified the importance of investigating how the number of accepting computation paths affects the entire behaviors of nondeterministic computation, leading to
Valiant's \cite{Val75,Val79} notion of \emph{counting functions}, which
concerns the number of accepting computation paths of nondeterministic Turing machines (or TMs) on given inputs.

The centerpiece of the subsequent research has pivoted around two crucial notions of \emph{counting function} that computes the total number of accepting computation paths of each run of an underlying nondeterministic TM \cite{Val75,Val79}, and \emph{gap function} that computes  the difference between the number of accepting computation paths and that of rejecting computation paths produced by nondeterministic TMs \cite{FFK94}.
These counting and gap functions have played an essential role in capturing various complexity classes.
The reader may refer to, e.g., \cite{For98} for those counting complexity classes.
The study of ``counting'' is thus deeply rooted at the nature of nondeterminism, and thus
decision problems characterized in terms of counting functions seem to tailor the essence of nondeterminism.

A main purpose of this  work is therefore to explore the roles of ``counting'' within the framework of nonuniform (polynomial) state complexity and to further enrich this research field initiated by Sakoda and Sipser for fully understanding the nature of nondeterminism.

\subsection{Exploitation of Close Connections to One-Tape Linear-Time Machines}\label{sec:exploitation}

The most important discovery of this work is a possible exploitation of a close connection between one-way finite automata and \emph{one-tape linear-time TMs}\footnote{As discussed in \cite{TYL10}, a one-tape linear-time nondeterministic/probabilistic/quantum TM is extremely sensitive to the definition of ``runtime''.
Here, the runtime of a nondeterministic TM means the length of any longest computation path. This definition is called the \emph{strong definition} of runtime \cite{TYL10}.}
in order to verify some of the separations among complexity classes of nonuniform families of finite automata.
Earlier, Hennie \cite{Hen65} and Tadaki, Yamakami, and Li \cite{TYL10} demonstrated close connections between one-way finite automata and one-tape linear-time TMs.
Hennie's fundamental result, for instance, sates that one-tape linear-time (more strongly, $o(n\log{n})$-time) deterministic TMs recognize
only regular languages \cite{Hen65}.
Using one-tape linear-time model of nondeterministic TMs, Tadaki, Yamakami, and Li \cite{TYL10} in 2004 initiated a study on counting and gap functions whose collections are respectively denoted by $\mathrm{1\mbox{-}\#LIN}$ and $\mathrm{1\mbox{-}GapLIN}$, where the prefix {``1-''} stands for ``one-tape'' and the suffix ``LIN'' does for ``linear time''.
These function classes  can help us characterize numerous one-tape linear-time counting complexity classes of decision problems, such as  $\mathrm{1\mbox{-}DLIN}$ (by deterministic TMs),  $\mathrm{1\mbox{-}NLIN}$ (by nondeterministic TMs),  $\mathrm{1\mbox{-}ULIN}$\footnote{This unambiguous complexity class was not formally introduced in \cite{TYL10,Yam10} but its properties were discussed there in connection to the inverse of functions.} (by unambiguous TMs), $1\mbox{-}\!\oplus\!\mathrm{LIN}$ (by parity\footnote{A \emph{parity} machine is a nondeterministic machine that accepts exactly when there are an odd number of accepting computation paths on each input.} TMs), $\mathrm{1\mbox{-}C_{=}LIN}$ (by exact counting\footnote{A nondeterministic machine is said to be \emph{exact counting} if it accepts an input exactly when there are the same numbers of accepting and rejecting computation paths.} TMs), $\mathrm{1\mbox{-}SPLIN}$ (by stoic probabilistic\footnote{A probabilistic machine that is viewed as a ``nondeterministic'' machine is called  \emph{stoic} if the machine accepts (resp., rejects) an input when  the number of accepting computation paths is one more than the number of rejecting computation paths (resp., all halting computation paths are rejecting) \cite{FFK94}.} TMs), and $\mathrm{1\mbox{-}PLIN}$ (by unbounded-error probabilistic TMs) \cite{TYL10,Yam10}.
These complexity classes $1\mbox{-}\!\oplus\!\mathrm{LIN}$, $\mathrm{1\mbox{-}C_{=}LIN}$, $\mathrm{1\mbox{-}SPLIN}$, and $\mathrm{1\mbox{-}PLIN}$ turn out to possess quite distinctive characteristics.


\begin{figure}[t]
\centering
\includegraphics*[height=4.3cm]{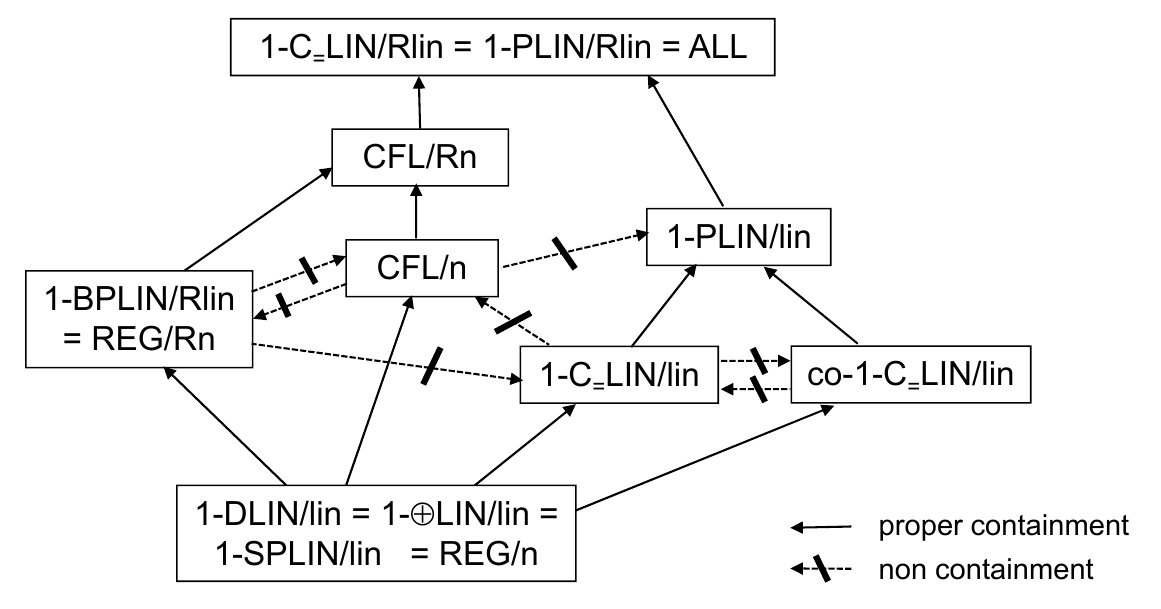}
\caption{Containments and separations of ``advised'' one-tape linear-time complexity classes shown in \cite{TYL10,Yam10,Yam11,Yam16}. The term ``ALL'' refers to the family of all languages. The notation $1\mbox{-}\mathrm{BPLIN}$ is the bounded-error variant of $1\mbox{-}\mathrm{PLIN}$ \cite{TYL10}. The complexity classes $\reg$/n \cite{TYL10} and $\cfl$/n \cite{Yam10} are the collections of all languages recognized respectively by 1dfa's and 1npda's with size-$n$ advice. The suffixes ``/Rn'' and ``/Rlin'' indicate the use of ``randomized'' advice, which is a distribution of advice strings, opposed to ``deterministic'' advice.}\label{fig:one-tape-figure}
\end{figure}


Nonuniformity can also be implemented by the use of ``advice'', which is an external source of information, to enhance the computational power of underlying machines.
Such advice is provided in the form of \emph{advice strings} to those machines \emph{in parallel to} standard input strings.
Finite automata and pushdown automata equipped with such (Karp-Lipton \cite{KL82} style)  ``advice'' were studied in a series of intensive research  \cite{DH95,TYL10, Yam08,Yam10,Yam11,Yam16}.
It turns out that the additional information contributes to the significant empowerment of underlying machines.

The distinctive role of linear-size (Karp-Lipton style) advice for one-tape linear-time TMs was further investigated in \cite{TYL10,Yam10}.
The supplemental use of such advice naturally introduces two nonuniform variants of $\mathrm{1\mbox{-}\#LIN}$ and $\mathrm{1\mbox{-}GapLIN}$, known  as $\mathrm{1\mbox{-}\#LIN/lin}$ and $\mathrm{1\mbox{-}GapLIN/lin}$, where the suffix ``/lin'' refers to the use of ``deterministic linear-size (Karp-Lipton style) advice.''
Similarly, the linear-size advised variants of the complexity classes $\mathrm{1\mbox{-}DLIN}$, $\mathrm{1\mbox{-}C_{=}LIN}$, $1\mbox{-}\!\oplus\!\mathrm{LIN}$, $\mathrm{1\mbox{-}SPLIN}$, and $\mathrm{1\mbox{-}PLIN}$ are expressed respectively as $\mathrm{1\mbox{-}DLIN/lin}$, $\mathrm{1\mbox{-}C_{=}LIN/lin}$, $1\mbox{-}\!\oplus\!\mathrm{LIN/lin}$, $\mathrm{1\mbox{-}SPLIN/lin}$, and $\mathrm{1\mbox{-}PLIN/lin}$. Figure~\ref{fig:one-tape-figure} summarizes known containments and separations among those one-tape linear-time complexity classes.

\subsection{Counting by Nonuniform Automata Families}\label{sec:state-complexity}

This work first formulates in Section \ref{sec:def-counting-func} the key function classes, respectively called $\onesharp$ and $\onegap$, of counting and gap (partial) functions.
More precisely, a counting function family in $\onesharp$ is an infinite series $\{(f_n,D_n)\}_{n\in\nat}$ consisting of partial (counting) functions $f_n$ producing the total number of accepting computation paths of 1nfa's  defined only on domains $D_n$. As for a gap function family $\{(g_n,D_n)\}_{n\in\nat}$ in $\onegap$, each partial (gap) function $g_n$  produces the difference between the total numbers of accepting and rejecting computation paths of a 1nfa on a domain $D_n$.
Here, a partial function $(f_n,D_n)$ can be viewed as a natural generalization of a promise decision problem $(L^{(+)}_n,L^{(-)}_n)$ used for $\oned$ and $\onen$.
Remember that, unlike $\mathrm{1\mbox{-}\#LIN/lin}$ and $\mathrm{1\mbox{-}GapLIN/lin}$, $\onesharp$ and $\onegap$ are composed of nonuniform families of ``partial'' functions whose values are ``defined'' only on predetermined domains.
In comparison to $\onesharp$ and $\onegap$, we also consider $\onef$, which consists of partial functions computable by nonuniform families of \emph{one-way deterministic finite transducers} (or 1dft's) with polynomially many inner states on specified domains.

With the use of such partial function families in $\onesharp$ and $\onegap$,
we can characterize various nonuniform polynomial state complexity classes:
$\oneu$ (by unambiguous automata families), $\onep$ (by unbounded-error probabilistic automata families), $\oneparity$ (by parity automata families), $\onecequal$ (by exact counting automata families), and $\onesp$ (by stoic probabilistic automata families),
following the naming practice\footnote{In \cite{SS78}, the terminology of $\oned$ and $\onen$ comes from 1-way Deterministic finite automata families and 1-way Nondeterministic finite automata families, respectively.} of Sakoda and Sipser.
For their precise definitions, refer to Section \ref{sec:def-complexity-class}.

\subsection{Main Contributions of This Work}

As the main contributions of this work, we prove in Sections \ref{sec:relationship} and \ref{sec:relation-pushdown} the containments and separations among the aforementioned nonuniform (polynomial) state complexity classes.
A key to the proofs of two main theorems (Theorems \ref{N-vs-cequal} and \ref{parity-vs-onep}) of this work is an exploitation of the aforementioned  close connection between nonuniform families of finite automata and one-tape linear-time TMs with linear-size (Karp-Lipton style) advice, which is stated in Lemmas \ref{property-1cequal} and \ref{property-onep}.
Many proofs of separation results among complexity classes exploit the fundamental properties of 1nfa's, in particular, one-way tape head move and the absence of memory space.

We further compare the complexity classes $\oneu$, $\onen$, and $\onep$ with the complexity classes $\onedpd$ and $\onenpd$  induced respectively by one-way deterministic and nondeterministic pushdown automata families \cite{Yam21a}.
In summary, Figure \ref{fig:class-separations} depicts the relationships among those nonuniform state complexity classes, shown in this work.
These relationships are quite different from those of the corresponding one-tape linear-time complexity classes, shown in Figure \ref{fig:one-tape-figure}.


In Section \ref{sec:functional-operation}, we further look into structural properties of $\onesharp$ and $\onegap$; in particular, closure properties under various functional operations (such as addition, subtraction, multiplication, and division). These closure properties were first studied by Ogiwara and Hemachandra \cite{OH93} for $\sharpp$ and $\gapp$. These closure properties of function classes are closely related to collapses of complexity classes of decision problems.
By exploiting such relationships, we use the separation results of Section \ref{sec:relationship} among counting complexity classes of promise decision problems  to verify the desired non-closure properties of $\onesharp$.



\begin{figure}[t]
\centering
\includegraphics*[height=4.8cm]{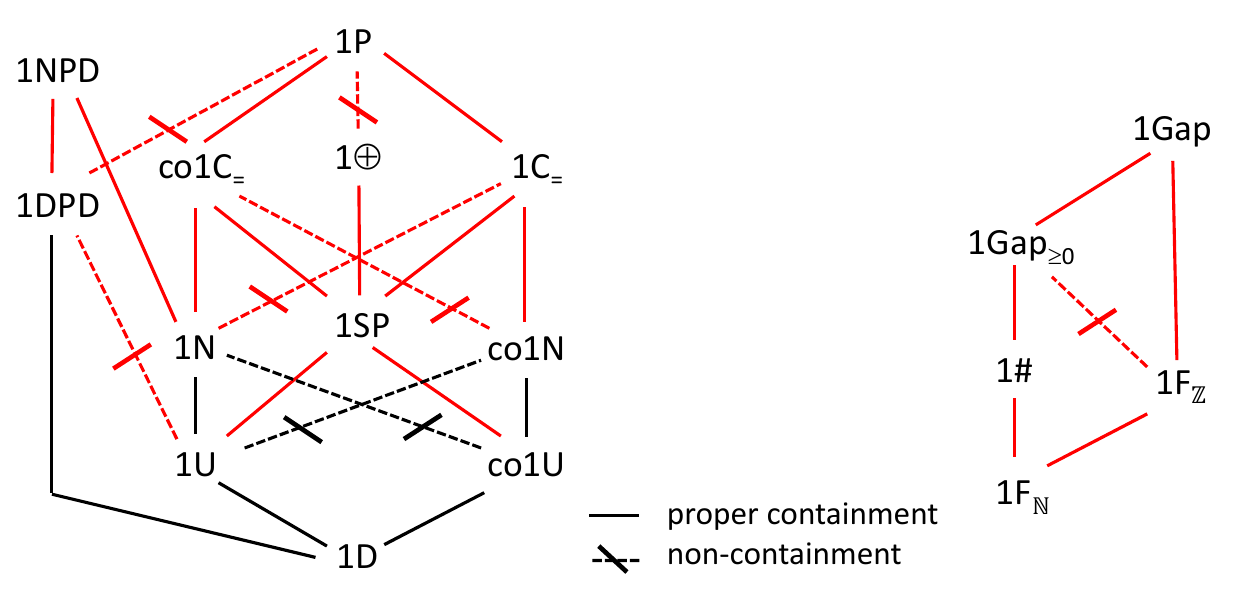}
\caption{Containments and separations of nonuniform polynomial state complexity classes shown in this work, marked by red lines and red dotted lines, except for $\oned\neq\oneu\neq\onen$ \cite{Yam22b,Yam22c}, $\co\oneu\nsubseteq\onen$ \cite{Yam22c}, and $\oned\neq\onedpd$ \cite{Yam21a} marked by black lines and black dotted lines.}\label{fig:class-separations}
\end{figure}


\section{Preparation: Notions and Notation}\label{sec:preparation}

We begin with explaining a number of basic but important notions and notation used in the subsequent sections.

\subsection{Numbers, Strings, and Promise Decision Problems}\label{sec:numbers}

The notation $\integer$ indicates the set of all integers.
All nonnegative integers are called \emph{natural numbers} and
form a unique set denoted by $\nat$. We also use $\nat^{+}$ to express
$\nat-\{0\}$.
For two integers $m$ and $n$ with $m\leq n$, the notation $[m,n]_{\integer}$ indicates the \emph{integer set} $\{m,m+1,m+2,\ldots,n\}$.
In particular,  we abbreviate $[1,n]_{\integer}$ as $[n]$
whenever $n\geq1$.
Given a set $S$,  $\PP(S)$ indicates the \emph{power set} of $S$.

We assume the reader's familiarity with automata theory. Let $\Sigma$
denote an \emph{alphabet}, which is a finite nonempty set of ``symbols'' or
``letters''. A finite sequence $s$ of symbols in $\Sigma$ is called a
\emph{string} over $\Sigma$ and its \emph{length} $|s|$ is the total number of symbols in $s$.
In this work, the \emph{empty string} (i.e., the string of length $0$) is always denoted by $\lambda$.
We use the notation $\Sigma^*$ to denote the set of all strings over $\Sigma$.
A \emph{language} over $\Sigma$ is simply a subset of $\Sigma^*$.
For any number $n\in\nat$, $\Sigma^n$ (resp., $\Sigma^{\leq n}$) denotes the set of all strings of length exactly $n$ (resp., at most $n$). Given a string $w\in\Sigma^*$ and a symbol $\sigma\in\Sigma$, $\#_{\sigma}(w)$ denotes the total number of occurrences of $\sigma$ in $w$.

A \emph{promise (decision) problem} over alphabet $\Sigma$ is
a pair $(L^{(+)},L^{(-)})$ of sets satisfying that
$L^{(+)}\cup L^{(-)} \subseteq \Sigma^*$ and $L^{(+)}\cap L^{(-)} =\setempty$, where instances in $L^{(+)}$ are generally called \emph{positive} (or YES instances) and those in $L^{(-)}$ are \emph{negative} (or NO instances).
All strings (or instances) in $L^{(+)}\cup L^{(-)}$ are customarily said to be \emph{valid} or \emph{promised}. A
language can be viewed as a special case of a promise problem of the form $(L,\overline{L})$, where $\overline{L}$ is the \emph{complement} of $L$ (i.e., $\overline{L}=\Sigma^*-L$).

To treat integers as binary strings, we use the following simple encoding schemes. For any positive integer $k$, the notation $bin(k)$ denotes the binary representation of $k$ leading with  $1$. Additionally, we set $bin(0)$ to be the empty string $\lambda$.
We further translate integers $k$ to binary strings $trans(k)$ as follows: $0$ is translated into $trans(0)=\lambda$,
a positive integer $k$ is translated into $trans(k)=1bin(k)$, and a negative integer $-k$ for $k>0$ is interpreted as $trans(-k)=0bin(k)$ in Section \ref{sec:functional-operation}.

\subsection{Various Types of One-Way Finite Automata}\label{sec:various-type}

In association with the main theme of ``counting'', we use various finite automata models.
Following \cite{Yam22a,Yam22b}, all
finite automata in this work are restricted so that they
should move their tape heads only in one
direction \emph{without making any $\lambda$-move},
where a $\lambda$-move refers to a step in which the tape head scans no input symbol
but it may change the inner state of the automaton.
Such machines are briefly called \emph{one way} machines in this work.
A \emph{one-way nondeterministic finite
automaton} (or a 1nfa, for short) is of the form
$(Q,\Sigma,\{\rhd,\lhd\},\delta,q_0,Q_{acc},Q_{rej})$, where $Q$
is a finite set of inner states, $\Sigma$ is an (input) alphabet,
$\rhd$ and $\lhd$ are two endmarkers,\footnote{We remark that automata  models with endmarkers and automata models with no endmarkers in general have different state complexities. Such a difference seems quite significant in the case of, e.g., one-way pushdown automata. See \cite{Yam21b} for a detailed discussion.}
$\delta$ is a transition function\footnote{This definition of $\delta$ clearly indicates the exclusion of any $\lambda$-move. When $\lambda$-moves are allowed, the corresponding automaton is called a 1.5nfa (as well as 1.5dfa) in \cite{Yam22a} in accordance with 1.5 quantum finite automata. For quantum finite automata, nevertheless, the use of $\lambda$-move causes a significant increase of computational power over the non-$\lambda$-move model. See, e.g., \cite{Yam22a}.}
mapping $(Q-Q_{halt}) \times\check{\Sigma}$ to $\PP(Q)$, where $\check{\Sigma}=\Sigma\cup\{\rhd,\lhd\}$, and
$Q_{acc}$ and $Q_{rej}$ are respectively sets of accepting states
and of rejecting states with $Q_{halt} = Q_{acc}\cup Q_{rej} \subseteq Q$ and $Q_{acc}\cap Q_{rej}=\setempty$.
We intentionally include $Q_{acc}$ and $Q_{rej}$
to the definition of 1nfa's.
An input (string) $x$ is given to an input tape, surrounded by the endmarkers $\rhd$ and $\lhd$. All tape cells are indexed by natural numbers, where $\rhd$ is located at cell $0$ and $\lhd$ is at cell $|x|+1$.
By the definition of $\delta$, when either $M$'s tape head moves off $\lhd$ or $M$ enters a halting state (i.e., either an accepting state or a rejecting state) on each computation path, $M$ is assumed to \emph{halt on this computation path}.
Since $M$ makes no $\lambda$-move, the use of the right endmarker $\lhd$ makes it possible that every computation path of $M$ must halt within $|x|+2$ steps for any input $x$.
We thus say that $M$ \emph{halts} on an input if $M$ starts on this input and eventually halts on all computation paths. For any input $x$, $M$ is said to \emph{accept} $x$ if there is an accepting computation path (i.e., a computation path ending with an accepting state), and $M$ is said to \emph{reject} $x$ if all computation paths end with rejecting states.
Moreover, we write $sc(M)$ for the
\emph{state complexity} $|Q|$ of $M$.
Given an input string $x$, the notation $M(x)$ stands for the ``outcome'' of $M$ on input $x$ whenever $M$ halts on $x$.
Given a promise problem $(L^{(+)},L^{(-)})$, a finite automaton $M$ is said to \emph{solve} it if $M$ accepts all instances in $L^{(+)}$ and rejects all instances in $L^{(-)}$; however, we do not impose any condition on invalid strings.
To express the total number of accepting computation paths of $M$ on $x$, we use the notation $\#M(x)$; by contrast, $\#\overline{M}(x)$ expresses the total  number
of rejecting computation paths of $M$ on $x$.

A \emph{one-way deterministic finite automaton} (or a 1dfa), in contrast,  uses a  transition function $\delta$ that maps $(Q-Q_{halt})\times \check{\Sigma}$ to $Q$.

Given a 1nfa $M$, we can modify it into another ``equivalent'' 1nfa $N$ so that  (1) $N$ makes exactly $c$ nondeterministic choices at every step and (2) $N$ produces exactly $c^{|\rhd x\lhd|}$ computation
paths on all inputs $x$ together with two endmarkers, where $c$ is an appropriately chosen constant in $\nat^{+}$.
For convenience, we call this specific form the \emph{branching normal form}.

\begin{lemma}[branching normal form]\label{branching-normal-form}
Let $M$ be any 1nfa solving a promise problem $(L^{(+)},L^{(-)})$.
(1) There exists another 1nfa $N$ in a branching normal form such that  $\#M(x)=\#N(x)$ for all valid $x$ and $sc(N)\leq 3 sc(M)$.
(2) There exists another 1nfa $N'$ in a branching normal form such that  $\#M(x)-\#\overline{M}(x) = \#N(x)-\#\overline{N}(x)$ for all valid $x$, and $sc(N)\leq 4 sc(M)$.
\end{lemma}

\begin{proof}
(1) Given a 1nfa $M$ as in the premise of the lemma, we wish to modify it to the desired 1nfa $N$ so that, for any input $x$, $N$ reads all $|x|+2$ input symbols (including the two endmarkers) written on an input tape before halting.
For this purpose, let $M= (Q,\Sigma,\{\rhd,\lhd\},\delta,q_0,Q_{acc},Q_{rej})$ and let $c$ denote the maximum number of nondeterministic choices of $M$.
We prepare $2c$ new inner states $p_1,p_2, \ldots,p_c, p_{rej,1},p_{rej,2}, \ldots,p_{rej,c}$ other than $Q$, where we treat each $p_{rej,i}$ as an additional rejecting state for every index $i\in[c]$. Formally, we set $Q'_{acc}=Q_{acc}$, $Q'_{rej}=Q_{rej}\cup
\{p_{rej,1},\ldots,p_{rej,c}\}$, and $Q'=Q\cup \{p_1,\ldots,p_c\}\cup \{p_{rej,1},\ldots,p_{rej,c}\}$ for $N$.
Note that $sc(N)=|Q|+2c \leq 3 sc(M)$ since $c\leq |Q|$.

For each symbol $\sigma\in\Sigma$ and any inner state $q\in Q$, we set $\delta'(q,\sigma)=\delta(q,\sigma)\cup\{p_1,p_2,\ldots,p_m\}$,
where $m=c-|\delta(q,\sigma)|$, and $\delta'(p_i,\sigma)=\{p_1,p_2,\ldots,p_c\}$, and $\delta'(p_i,\lhd)=\{p_{rej,1},p_{rej,2},\ldots,p_{rej,c}\}$ for any $i\in[c]$.
Clearly, $N$ makes exactly $c$ nondeterministic choices at every step. By the above definition of $N$, we also obtain $\#M(x)=\#N(x)$ for all inputs $x$.

(2) A basic idea is similar to (1). Here, instead of using $c$ itself, we take the minimal even number $d$ satisfying $d\geq c$ and prepare $3d$ new inner states $p_1,\ldots,p_d, p_{acc,1},\ldots,p_{acc,d}, p_{rej,1},\ldots,p_{rej,d}$ other than $Q$. Note that $sc(N)=|Q|+3d\leq 4 sc(M)$. In particular, we need to define  $\delta'(p_i,\lhd)=\{p_{acc,1},\ldots,p_{acc,d/2}, p_{rej,1},\ldots,p_{rej,d/2}\}$ for all $i\in[d]$. This implies that each inner state $p_i$ leads to an equal number of accepting and rejecting computation paths. It thus follows that  $\#M(x)-\#\overline{M}(x)$ coincides with $\#N(x)-\#\overline{N}(x)$ for all inputs $x$.
\end{proof}


The following property of 1nfa's is useful for proving Lemma \ref{cequal-closure} and Proposition \ref{onecequal-vs-onep}.

\begin{lemma}\label{cequal-path-number}
For every 1nfa $M$, there exists another 1nfa $N$ such that, for any $x$, if $\#M(x)=\#\overline{M}(x)$, then $\#N(x)=\#\overline{N}(x)$, and if $\#M(x)\neq \#\overline{M}(x)$, then $\#N(x) > \#\overline{N}(x)$. Here, the inequality $>$ can be replaced by $<$ by swapping between accepting states and ejecting states. Moreover, $sc(N)\leq sc(M)^2$ holds.
\end{lemma}

\begin{proof}
From a given 1nfa $M$ with a set $Q_M$ of inner states, we design a new 1nfa $N$ equipped with a set $Q_N$ of inner states as follows. Let  $Q_{M,acc}=\{a_1,\ldots,a_{m_1}\}$ and $Q_{M,rej}=\{r_1,\ldots,r_{m_2}\}$ for two numbers $m_1,m_2\in\nat^{+}$.
For simplicity, we assume that $M$ always halts just after reading $\lhd$.
The desired 1nfa $N$ performs a simultaneous simulation of $M$'s two separate computation paths composed of two separate series of inner states by maintaining two separate elements of these series as a single inner state of $N$.
More formally, we set $Q_N=\{(q_1,q_2)\mid q_1,q_2\in Q_M\}$, $Q_{N,acc}=\{(a_i,a_j)\mid i,j\in[m_1]\}\cup \{(r_i,r_j)\mid i,j\in[m_2]\}$, and $Q_{N,rej}=\{(a_i,r_j)\mid i\in[m_1],j\in[m_2]\} \cup \{(r_i,a_j)\mid i\in[m_2],j\in[m_1]\}$. Since $|Q_{N}|\leq |Q_{M}|^2$, we obtain $sc(N)\leq sc(M)^2$.

If two computation paths end in the same type of halting states (i.e., either accepting states or rejecting states), then $N$ enters its own accepting states. Otherwise, $N$ enters its own rejecting states.
The behavior of $N$ is formally dictated in terms of its transition function $\delta_N$ defined by $\delta_N((q_1,q_2),\sigma) = \{(p_1,p_2)\mid p\in\delta_M(q_1,\sigma),p_2\in\delta_M(q_2,\sigma)\}$ for any $\sigma\in\check{\Sigma}$.
By the definitions of $Q_{N,acc}$ and $Q_{N,rej}$, it follows that $\#N(x)$ equals $\#M(x)^2+\#\overline{M}(x)^2$ and that $\#\overline{N}(x)$ equals $2\#M(x)\cdot \#\overline{M}(x)$. From these equations, we instantly obtain $\#N(x)-\#\overline{N}(x)= (\#M(x)-\#\overline{M}(x))^2$.
It then follows that, for any valid string $x$, $\#M(x)=\#\overline{M}(x)$ implies $\#N(x)=\#\overline{N}(x)$. On the contrary, if $\#M(x)\neq \#\overline{M}(x)$, then $\#N(x)>\#\overline{N}(x)$ follows.
\end{proof}


When a finite automaton is further equipped with a \emph{write-once}\footnote{A \emph{write-once tape} means that its tape head always moves to the next blank cell after the tape head writes a non-blank symbol.} output tape, we call it a \emph{finite transducer} (which can compute ``functions'') to distinguish it from the aforementioned finite automata acting as \emph{language acceptors}.
In this work, we consider only one-way deterministic finite transducers (or 1dft's, for short) equipped with transition functions $\delta$ mapping $(Q-Q_{halt})\times \check{\Sigma}$ to $Q\times\Gamma^*$ for two alphabets $\Sigma$ and $\Gamma$. Since a 1dft does not need to ``accept'' or ``reject'' an input, the 1dft requires only a set $Q_{halt}$ of halting states.
We also remark that, since a 1dft always moves its input-tape head to the right, we need to allow the 1dft to write multiple symbols (that is, a ``string'') onto its output tape in a single step by moving its tape head over multiple tape cells. The \emph{single-step output size} refers to the maximum number of symbols written down on the output tape in a single step.


A \emph{one-way nondeterministic pushdown automaton} (or a 1npda, for short) $M$ is a nonuple $(Q,\Sigma,\{\rhd,\lhd\}, \Gamma,\delta,q_0,\bot, Q_{acc}, Q_{rej})$ similar to a 1nfa but additionally equipped with a stack alphabet $\Gamma$ and a bottom marker $\bot$ for a stack, where $\delta$ maps $(Q-Q_{halt})\times \check{\Sigma}_{\lambda}\times \Gamma$ to $\PP(Q\times \Gamma^{*})$, where $\check{\Sigma}_{\lambda} = \check{\Sigma}\cup\{\lambda\}$.
A 1npda starts with $\rhd{x}\lhd$ ($x\in\Sigma^*$) written on its input tape and $\bot$ (bottom marker) on its stack. The (input) tape head is initially located at the leftmost tape cell and the stack head is always scanning the topmost stack cell.
A transition ``$(p,\gamma)\in \delta(q,\sigma,a)$'' means that, when reading $\sigma$ on the input tape and $a$ at the top cell of the stack, $M$ changes its inner state from $q$ to $p$ and replaces $a$ by $\gamma$ on the stack in a single step. Whenever $\sigma\neq \lambda$, the tape head must move to the right.
When $\sigma=\lambda$, on the contrary, $M$ makes a $\lambda$-move (i.e., the tape head stays still). We demand that the bottom marker is not pushed or popped at any step.
The \emph{push size} $e$ of $M$ is defined to be $\max_{(p,q,\sigma,a)}\{|\gamma|: (p,\gamma)\in \delta(q,\sigma,a)\}$ with $p,q\in Q$, $\sigma\in\check{\Sigma}_{\lambda}$, and $a\in\Gamma$. Hence, the range of $\delta$ can be restricted to $\PP(Q\times \Gamma^{\leq e})$.
The \emph{stack content} refers to the string in $(\Gamma^{(-)})^*\bot$ stored in the stack, where $\Gamma^{(-)}=\Gamma-\{\bot\}$. The \emph{stack height} is the length of this stack content. A \emph{stack height history} of a computation path of $M$ is a series of stack heights obtained as the tape head moves to the right following this computation path.
The \emph{stack-state complexity} of $M$, denoted by $ssc(M)$, is defined to be $|Q||\Gamma^{\leq e}|$. The fundamental notions and notation of 1nfa's are used also for 1npda's.

A deterministic variant of a 1npda is called a \emph{one-way deterministic pushdown automaton} (or a 1dpda) if the condition  $|\delta(q,\sigma,a)\cup\delta(q,\lambda,a)|\leq1$ holds for all
triplets $(q,\sigma,a)\in Q\times \check{\Sigma}\times \Gamma$.

\subsection{Nonuniform Families of Promise Problems and Finite Automata}\label{sec:nonuniform-family}

Given a fixed alphabet
$\Sigma$, we consider a family
$\LL=\{(L_n^{(+)},L_n^{(-)})\}_{n\in\nat}$ of promise problems,
each $(L_n^{(+)},L_n^{(-)})$ of which indexed by $n$ is a promise problem  over $\Sigma$. The set $L_n^{(+)}\cup L_n^{(-)}$ of all promised (or valid) instances is succinctly denoted by $L_n$. Here, instances in $L_n$ may not be limited to length-$n$ strings. It is important to note that the underlying alphabet $\Sigma$ is fixed for all promise problems $(L_n^{(+)},L_n^{(-)})$ in $\LL$ in accordance with the setting of \cite{Yam19a,Yam19b,Yam21a,Yam22a,Yam22b}.

The \emph{complement} of $\LL$  is $\{(L_n^{(-)},L_n^{(+)})\}_{n\in\nat}$ and is denoted $\co\LL$.
Given two families $\LL=\{(L_n^{(+)},L_n^{(-)})\}_{n\in\nat}$ and $\KK=\{(K_n^{(+)},K_n^{(-)})\}_{n\in\nat}$, we define the \emph{intersection} $\LL\cap \KK$ to be $\{(L_n^{(+)}\cap K_n^{(+)}, L_n^{(-)}\cup K_n^{(-)})\}_{n\in\nat}$ and the \emph{union} $\LL\cup\KK$ to be $\{(L_n^{(+)}\cup K_n^{(+)}, L_n^{(-)}\cap K_n^{(-)})\}_{n\in\nat}$.

To solve (or recognize) $\LL$, we use  a nonuniform family $\MM = \{M_n\}_{n\in\nat}$ of finite automata indexed by natural numbers.  Formally, we say that $\MM$ \emph{solves} (or \emph{recognizes}) $\LL$ if, for any index  $n\in\nat$, $M_n$ solves $(L_n^{(+)},L_n^{(-)})$.
We say that $\MM$ is of \emph{polynomial size} if there is a polynomial $p$ satisfying $sc(M_n)\leq p(n)$ for all $n\in\nat$.
The complexity class $\oned$ (resp., $\onen$) consists of all families of promise problems solvable by nonuniform families of polynomial-size 1dfa's (resp., 1nfa's).
Similarly, a family $\{N_n\}_{n\in\nat}$ of 1dpda's (resp., 1npda's) is said to be of \emph{polynomial size} if there exists a polynomial $p$ satisfying $ssc(N_n)=|Q_n||\Gamma_n^{\leq e_n}|\leq p(n)$ for all $n\in\nat$, where $e_n$ is the push size of $N_n$.
We note that $\Sigma$ is fixed independent of index $n$ but $\Gamma_n$ is allowed to vary according to $n$.
We then define two complexity classes $\onedpd$ and $\onenpd$  using polynomial-size 1dpda's and 1npda's, respectively \cite{Yam21a}.

We further consider nonuniform families of \emph{partial functions}. Let $f$ denote a partial function from $\Sigma^*$ to $\Gamma^*$ for two alphabets $\Sigma$ and $\Gamma$, where ``partial'' means that $f$ is treated as ``defined''  on its domain $D$, which is a subset of $\Sigma^*$, and $f$ is treated as ``undefined'' on $\Sigma^*-D$.
To emphasize the domain of $f$, we intend to write $(f,D)$ instead of $f$.  Given two alphabets $\Sigma$ and $\Gamma$, we express  a family $\FF$ of partial functions mapping $\Sigma^*$ to $\Gamma^*$ as  $\{(f_n,D_n)\}_{n\in\nat}$.

A family $\{M_n\}_{n\in\nat}$ of 1dft's is said to have \emph{polynomial size} if there are a polynomial $p$ and a constant $c>0$ such that (i)  $|Q_n|\leq p(n)$ and (ii) the single-step output size of $M_n$ is upper-bounded by $c\log{n}+c$ for any $n\in\nat$. Moreover, we say that $\{M_n\}_{n\in\nat}$
\emph{computes} $\FF=\{(f_n,D_n)\}_{n\in\nat}$ if, for any $n\in\nat$ and for all $x\in D_n$, $M_n$ on input $x$ produces $f_n(x)$ on its write-once output tape. Remark that, even if $x\notin D_n$, $M_n$ may possibly write an arbitrary string on its output tape but we invalidate such a string.

The notation $\onef$ denotes the collection of all families $\{(f_n,D_n)\}_{n\in\nat}$ of partial functions computed by certain families $\{M_n\}_{n\in\nat}$ of polynomial-size 1dft's.

\section{Introduction of Counting Functions and Counting Complexity Classes}\label{sec:counting-class}

Let us introduce two main function classes, $\onesharp$ and $\onegap$, composed of nonuniform families of \emph{partial functions} that are ``witnessed'' by polynomial-size 1nfa's.

\subsection{Definitions of Counting and Gap Functions}\label{sec:def-counting-func}

Following Valiant's \cite{Val75,Val79} definition of the counting function class $\sharpp$ and the gap function class $\gapp$ of Fenner, Fortnow, and Kurtz \cite{FFK94}, we wish to formulate two counting and gap function classes $\onesharp$ and $\onegap$ of our interest using nonuniform families of 1nfa's.

\begin{definition}
(1) The counting function class $\onesharp$ (pronounced ``one sharp'' or ``one pound'') is defined to be the collection of all families  $\{(f_n,D_n)\}_{n\in\nat}$ of partial functions such that there exists a nonuniform family $\{M_n\}_{n\in\nat}$ of polynomial-size 1nfa's for which $f_n(x)=\#M_n(x)$ for all $n\in\nat$ and all $x\in D_n$.

(2) In contrast, we define $\onegap$ to be the collection of all families $\{(f_n,D_n)\}_{n\in\nat}$ of partial functions characterized as $f_n(x)= \#M_n(x) - \#\overline{M}_n(x)$ for all $n\in\nat$ and all $x\in D_n$ for an appropriately chosen nonuniform family $\{M_n\}_{n\in\nat}$ of polynomial-size 1nfa's.
Additionally, the function class $\onegap_{\geq0}$ is composed of all partial function families in $\onegap$ that take only nonnegative integer values unless their outputs are undefined.
\end{definition}

For the sake of the reader, we quickly review two simple examples.

\begin{example}
Given a string $x\in\{0,1,2\}^*$, we define $f_n(x)$ to be the total number of $1$s in $x$ and $g_n(x)$ to be the total number of $1$s minus the total number of $0$s in $x$. We then claim that $\{(f_n,\Sigma^n)\}_{n\in\nat}$ and $\{(g_n,\Sigma^n)\}_{n\in\nat}$ are respectively in $\onesharp$ and $\onegap$. To prove that $\{(f_n,\Sigma^n)\}_{n\in\nat}\in\onesharp$, let us consider the 1nfa's $M_n$ that behaves as follows.
On input $x$, nondeterministically choose a number $i\in[|x|]$ and read the $i$th tape cell (by moving a tape head from $\rhd$ to the right and making a nondeterministic choice of either reading a tape symbol or skipping the tape cell). If the content of this tape cell is $1$, then accept $x$; otherwise, reject  $x$. It follows that $\#M_n(x)$ equals the total number of 1s in $x$, and thus we obtain $f_n(x)=\#M_n(x)$ for all $x$.
To see that  $\{(g_n,\Sigma^n)\}_{n\in\nat}\in\onegap$, in contrast, we slightly modify $M_n$ to $N_n$ as follows. Nondeterministically choose a number $i\in[|x|]$ and read the $i$th tape cell. If its content is $1$ (resp., $0$), then accept (resp., reject) $x$. When $2$ is found in the tape cell, branch out into two transitions with one accepting state and one rejecting state. It is not difficult to see that $\#N_n(x)-\#\overline{N}_n(x)$ equals $g_n(x)$.
\end{example}


Before stating the second example, we explain the special notation $[i_1,i_2,\ldots,i_k]$.
For any $k\in\nat^{+}$ and any $i_1,i_2,\ldots,i_k\in\nat^{+}$,  $[i_1,i_2,\ldots,i_k]$ denotes the binary string of the form $1^{i_1}01^{i_2}0\cdots 01^{i_k}0$.
Given any $n\in\nat^{+}$, let $A_n$ denote the collection of all such strings $[i_1,i_2,\ldots,i_k]$ with $k\in\nat^{+}$ and $i_1,i_2,\ldots,i_k\in[n]$.
We further define $A_{\infty}$ to be the set $\bigcup_{n\in\nat^+}A_n$.
For any $r=[i_1,i_2,\ldots,i_k]\in A_{\infty}$, $Set(r)$ denotes the set $\{i_1,i_2,\ldots,i_k\}$. In comparison, another notation $MSet(r)$ is used for the corresponding \emph{multi set}. Obviously, it follows that $|Set(r)|\leq |MSet(r)|=k$.
For any index $e\in[k]$, the notation $(r)_{(e)}$ indicates the $e$th entry $i_e$ in $r$.
Moreover, the set $A_n(m,k)$ consists of all strings of the form $[i_1,i_2,\ldots,i_k]$ with $i_1,i_2,\ldots,i_k\in [m]$. See \cite{Yam22b} for more information.


\begin{example}
As another concrete example of counting function families, let us consider $D_n=\{r_1\#r_2\mid r_1,r_2\in A_n\}$ and $f_n(x) = |\{e\in\nat^{+}\mid (r_1)_{(e)} \in Set(r_2)\}|$ for $x=r_1\# r_2$ in $D_n$.
We then show that the function family  $\{(f_n,D_n)\}_{n\in\nat}$ belongs to $\onesharp$ by considering the following 1nfa $M_n$ on any input $x$. Assume that $x$ has the form $r_1\# r_2$ for $r_1,r_2\in A_n$ with $r_1=[i_1,i_2,\ldots,i_k]$ and $r_2=[j_1,j_2,\ldots,j_l]$, where $i_m,j_m\in[n]$. The 1nfa $M_n$ nondeterministically chooses an index $e\in[n]$ (by moving its tape head to the right from $\rhd$ and making a nondeterministic choice of either reading a number $i_e$ or skipping it) and then reads $i_e$ ($=(r_1)_{(e)}$). Once reading $i_e$, $M_n$ moves its tape head to $\#$ and starts checking whether $i_e$ matches $j_m$ for a certain index $m\in[l]$ deterministically by moving the tape head to the right.
If there exists such an index $m$, then $M_n$ accepts; otherwise, $M_n$ rejects. Clearly, we obtain $f_n(x)=\#M_n(x)$. This entire description of $M_n$ requires only $O(n)$ inner states.
We further define $g_n(x) = |\{e\in \nat^{+}\mid (r_1)_{(e)} \in Set(r_2)\}| - |\{e\in \nat^{+}\mid (r_1)_{(e)} \notin Set(r_2)\}|$. Since $\#\overline{M}_n(x)$ matches $|\{e\in \nat^{+}\mid (r_1)_{(e)}\notin Set(r_2)\}|$, the equality $g_n(x)=\#M_n(x) - \#\overline{M}_n(x)$ follows. Therefore, $\{(g_n,D_n)\}_{n\in\nat}$ belongs to $\onegap$.
\end{example}

Let us recall from Section \ref{sec:nonuniform-family} that $\onef$ consists of all families of partial functions, mapping $\Sigma^*$ to $\Gamma^*$ for arbitrary alphabets $\Sigma$ and $\Gamma$, which are computed by nonuniform families of polynomial-size 1dft's. Here, we use the translation $trans(\cdot)$ between integers and binary strings described in Section \ref{sec:numbers}.
In accordance with the definitions of $\onesharp$ and $\onegap$,  $trans(\cdot)$ makes it possible to treat functions mapping $\Sigma^*$ to $\nat$ or $\integer$ as functions $f^{(trans)}$ mapping $\Sigma^*$ to $\{0,1\}^*$ defined by $f^{(trans)}(x) = trans(f(x))$ whenever $f(x)$ is defined.
As customary in computational complexity theory, we freely identify $f_n$ with $f_n^{(trans)}$, and thus we can define $\onef_{\integer}$ (resp., $\onef_{\nat}$) to be the class composed of all nonuniform families $\{(f_n,D_n)\}_{n\in\nat}$ of partial functions mapping $\Sigma^*$ to $\integer$ (resp., to $\nat$) for arbitrary alphabets $\Sigma$ under the condition that  $\{(f_n^{(trans)},D_n)\}_{n\in\nat}$ must belong to $\onef$.  Clearly, it follows that $\onef_{\nat}\subseteq \onef_{\integer}\subseteq \onef$.


In the end of this section, we briefly argue a direct connection to so-called \emph{weighted automata} \cite{Sch61}. Notice that a single 1nfa $M$ can introduce the counting function $f$ of the form $f(x)=\# M(x)$ for all $x$.
Let us consider the following weighted automaton $N$. The initial weight function assigns weight $1$ to an initial inner state, say, $q_0$.
For each transition of $N$ has weight $1$ except for any transition to a rejecting state of $M$, which has weight $0$. The weight of each computation path $\gamma$ is simply the product of the weights of all transitions made along $\gamma$. For each input $x$, $N$ finally produces the weight $w_N(x)$, which is the sum of the weights of all possible computation paths of $N$. It is clear that $f(x)$ matches $w_N(x)$. Thus, weighted automata can define counting functions.
On the contrary, a weight automaton $N$ with natural numbers can be simulated by the 1nfa $M$ that produces $w$ branches for each transition of $N$ with weight $w$. Moreover, the values of an initial weight function of $N$ must be simulated at the first move of $M$. Therefore, the number of computation paths of $M$ equals the weight of $N$.
See Section \ref{sec:discussion} for a relevant subject.

\subsection{Basic Relationships among 1F, 1\#, and 1Gap}\label{sec:basic-rel-among}

In Section \ref{sec:def-counting-func}, we have introduced five function classes, $\onef_{\nat}$, $\onef_{\integer}$, $\onesharp$, $\onegap$ and $\onegap_{\geq0}$, which consist of various families of counting and gap functions. Hereafter, we wish to discuss basic relationships, which are shown in Figure \ref{fig:class-separations}, among those function classes.
We remark that these basic relationships reflect similar relationships among $\fp$, $\sharpp$, and $\gapp$ (see, e.g., \cite{For98}).


\begin{lemma}\label{inclusion-1F-to-1sharp}
(1) $\onef_{\nat}\subsetneqq \onesharp \subseteq \onegap_{\geq0}$. (2) $\onef_{\integer} \nsubseteq \onegap_{\geq0}$ and $\onef_{\integer} \subseteq \onegap$.
\end{lemma}

All inclusion relations of this lemma are in fact \emph{strict}. However, we will delay the proofs of the separations of $\onesharp\neq\onegap_{\geq0}$ and $\onef_{\integer}\neq \onegap$ to Proposition \ref{separation-onef} in Section \ref{sec:implications} because they leverage separation results on $\onecequal$, $\co\onecequal$, $\oned$, and $\onep$.

\begin{proofsketchof}{Lemma \ref{inclusion-1F-to-1sharp}}
(1)
We first show the inclusion of $\onef_{\nat}\subseteq \onesharp$.
Given any family $\FF=\{(f_n,D_n)\}_{n\in\nat}$ of partial functions in $\onef_{\nat}$,
It suffices to construct a new 1nfa $H_n$, which produces exactly $f_n(x)$ accepting computation paths, namely, $f_n(x)=\#H_n(x)$. We remark that $M_n$ produces either $\lambda$ or $1bin(m)$ for a certain number $m\in\nat^{+}$ on its output tape.
Since the corresponding family $\{(f_n^{(trans)},D_n)\}_{n\in\nat}$ belongs to $\onef$, there is a family $\{M_n\}_{n\in\nat}$ of polynomial-size 1dft's computing it.
On input $x$, $H_n$ simulates $M_n$ on $x$ to obtain either $\lambda$ or $1bin(m)$ and then generates $m$ different computation paths when $\lambda$ is not an output.
Since $f_n(x)$ equals either $0$ or $m$, we force $H_n$ to have either $0$ or $m$ accepting computation paths.

We next show that $\onef_{\nat}\neq \onesharp$. It suffices to show that (*) $\onef_{\nat}=\onesharp$ implies $\oned=\onen$. Since $\oned\neq \onen$ \cite{SS78}, we immediately obtain $\onef_{\nat}\neq \onesharp$ from (*). To prove (*), we assume on the contrary that $\onef_{\nat}=\onesharp$.
Let $\LL=\{(L_n^{(+)},L_n^{(-)})\}_{n\in\nat}$ be any family in $\onen$ and take a family $\MM=\{M_n\}_{n\in\nat}$ of 1nfa's that solves $\LL$. Let us define $D_n=L_n^{(+)}\cup L_n^{(-)}$ and  $f_n(x)=\#M_n(x)$ for any $x$.
The function family $\FF=\{(f_n,D_n)\}_{n\in\nat}$ obviously belongs to $\onesharp$. Since $\onef_{\nat}=\onesharp$, $\FF$ belongs to $\onef_{\nat}$.
In other words,  $f_n(x)\geq1$ holds for all $x\in L_n^{(+)}$ and   $f_n(x)=0$ holds for all $x\in L_n^{(-)}$. We take a family $\NN=\{N_n\}_{n\in\nat}$ of polynomial-size 1dft's that computes $\{(f_n^{(trans)},D_n)\}_{n\in\nat}$. Consider the following 1dfa $K_n$. On input $x$, simulate $N_n$. Since $trans(0)=\lambda$, if $N_n$ writes at least one non-blank symbol on this output tape, then we accept; otherwise, we reject. This implies that $\LL$ is in $\oned$. We thus conclude that $\onen\subseteq \oned$.

To show that $\onesharp\subseteq \onegap_{\geq0}$, let $\FF=\{(f_n,D_n)\}_{n\in\nat}$ denote any family in $\onesharp$. Take a family  $\MM=\{M_n\}_{n\in\nat}$ of polynomial-size 1nfa's that witnesses ``$\FF\in\onesharp$''. Let us consider the following 1nfa's $N_n$. On input $x$, simulate $M_n$ on $x$.
On each computation path, if $M_n$ enters an accepting state, then $N_n$ does the same. Otherwise, $N_n$ branches into two computation paths and accepts on one path and rejects on the other. This implies that $\FF$ belongs to $\in \onegap_{\geq0}$.

(2) Since $\onef_{\integer}$ contains families of partial functions handling negative values, $\onef_{\integer}$ is not included in $\onegap_{\geq0}$.
For the inclusion $\onef_{\integer}\subseteq \onegap$, its proof is in essence similar to that of (1); however, if $f_n(x)\geq0$, then we produce exactly $f_n(x)$ more accepting computation paths than rejecting ones; otherwise,   exactly $f_n(x)$ more rejecting computation paths than accepting ones.
\end{proofsketchof}

Given two function classes $\UU$ and $\VV$, the special notation $\UU-\VV$ denotes the collection of families $\{(f_n,D_n)\}_{n\in\nat}$ such that there are two families $\GG=\{(g_n,E_n)\}_{n\in\nat}\in \UU$ and $\GG'=\{(g'_n,E'_n)\}_{n\in\nat}\in \VV$ for which $D_n=E_n\cap E'_n$ and $f_n = g_n - g'_n$ (i.e., $f_n(x)=g_n(x)-g'_n(x)$ for all $x\in D_n$) for any index  $n\in\nat$.
With this notation, $\onegap$ is succinctly expressed in terms of $\onef_{\nat}$, $\onesharp$, and $\onegap$ as follows.

\begin{lemma}\label{one-gap-character}
(1) $\onegap = \onesharp - \onesharp = \onesharp - \onef_{\nat} = \onef_{\nat} - \onesharp$. (2) $\onegap = \onegap - \onegap$.
\end{lemma}

\begin{proofsketch}
(1) We want to show that $\onegap\subseteq \onesharp-\onesharp$.
We begin with any family $\FF=\{(f_n,D_n)\}_{n\in\nat}$ in $\onegap$. Take a family $\{M_n\}_{n\in\nat}$ of polynomial-size 1nfa's witnessing the fact of ``$\FF\in\onegap$''. We define $N_n$ as $\overline{M}_n$ (which is obtained from $M_n$ by flipping accepting and rejecting states). Define $g_n(x)=\#M_n(x)$ and $h_n(x) =\#N_n(x)$ so that $f_n(x) = g_n(x) - h_n(x)$ holds for all strings $x\in D_n$.
Since $\GG=\{(g_n,D_n)\}_{n\in\nat}$ and $\HH=\{(h_n,D_n)\}_{n\in\nat}$ are both in $\onesharp$,
we conclude that $\FF\in\onesharp-\onesharp$.

To show that $\onesharp-\onesharp\subseteq \onesharp-\onef_{\nat}$, let us take two families $\GG=\{(g_n,D_n)\}_{n\in\nat}$ and $\GG'=\{(g'_n,D_n)\}_{n\in\nat}$ in $\onesharp$.
We define the partial function $h_n$ by setting $h_n=g_n-g'_n$. Note that $\HH = \{(h_n,D_n)\}_{n\in\nat}$ is in $\onesharp - \onesharp$.
There are two families $\{M_n\}_{n\in\nat}$ and $\{N_n\}_{n\in\nat}$ satisfying both $g_n(x)=\#M_n(x)$ and $g'_n(x)=\#N_n(x)$ for all $n\in\nat$ and all $x\in D_n$.
By Lemma \ref{branching-normal-form}, we can assume that $N_n$ makes exactly $c$ nondeterministic choices at every step until reading $\lhd$, where $c$ is an appropriate positive constant.
This instantly implies that $\#N_n(x)+\#\overline{N}_n(x)= c^{|\rhd x\lhd|}$ ($=c^{|x|+2}$).
By setting $f_n(x)= \#M_n(x)+\#\overline{N}_n(x)$, we can easily obtain $\{(f_n,D_n)\}_{n\in\nat}\in\onesharp$ (see also Lemma \ref{closure-onesharp}).
It is not difficult to verify that $h_n(x) = f_n(x) - c^{|x|+2}$, and thus
$\HH$ belongs to $\onesharp-\onef_{\nat}$.

With a similar idea, we can show that $\onesharp-\onesharp\subseteq \onef_{\nat}-\onesharp$.
Since $\onef_{\nat}\subseteq \onesharp$ by Lemma \ref{inclusion-1F-to-1sharp}(1), it follows that $\onesharp-\onef_{\nat} \subseteq \onesharp - \onesharp$ and $\onef_{\nat}-\onesharp \subseteq \onesharp - \onesharp$.

Finally, we want to verify that $\onesharp-\onesharp \subseteq \onegap$.
Choose two families $\{M_n\}_{n\in\nat}$ and $\{N_n\}_{n\in\nat}$ satisfying $f_n(x) = \#M_n(x)-\#N_n(x)$ for all $n\in\nat$ and all $x$ in a certain domain $D_n$. We define a new 1nfa $H_n$ in the following way. On input $x$, branch out into two computation paths. On one path, simulate $M_n$; on the other path, simulate $N_n$. If $M_n$ accepts, then accept the input. By contrast, if $M_n$ rejects $x$, then further branch out into two distinct configurations, one of which leads to an accepting state and the other leads to a rejecting state. A similar treatment is applied to $N_n$ but ``accept'' and ``reject'' should be exchanged. As a consequence, we can conclude that $\{(f_n,D_n)\}_{n\in\nat}$ belongs to $\onegap$.

(2) It is obvious that $\onegap\subseteq \onegap - \onegap$. To show the converse inclusion, let us consider a family $\HH=\{(h_n,D_n)\}_{n\in\nat}$ in $\onegap - \onegap$. There are two families $\{(f_n,D_n)\}_{n\in\nat}$ and $\{(g_n,D_n)\}_{n\in\nat}$ in $\onegap$ satisfying $h_n=f_n-g_n$ for all $n\in\nat$. From (1), these two families are further characterized by four families $\{(f_n^{(i)},D_n)\}_{n\in\nat}$ and $\{(g_n^{(i)},D_n)\}_{n\in\nat}$ for $i\in\{1,2\}$ in $\onesharp$ satisfying that $f_n=f_n^{(1)}-f_n^{(2)}$ and $g_n=g_n^{(1)}-g_n^{(2)}$ for all $n\in\nat$. We then obtain $f_n-g_n=(f_n^{(1)}-f_n^{(2)}) - (g_n^{(1)}-g_n^{(2)}) = (f_n^{(1)}+g_n^{(2)}) - (f_n^{(2)}+g_n^{(1)})$. It thus suffices to define $h_n^{(1)}= f_n^{(1)}+g_n^{(2)}$ and $h_n^{(2)}=f_n^{(2)}+g_n^{(1)}$.
By combining two $\onesharp$-function families, it is possible to show that both $\{(h_n^{(1)},D_n)\}_{n\in\nat}$ and $\{(h_n^{(2)},D_n)\}_{n\in\nat}$ belong to $\onesharp$ (see also Lemma \ref{closure-onesharp}). Since $h_n=h_n^{(1)}-h_n^{(2)}$ for all $n\in\nat$, we conclude by (1) that $\HH$ is indeed in $\onegap$.
\end{proofsketch}


We next examine closure properties of $\onesharp$ associated with \emph{homomorphisms}. Given two alphabets $\Sigma$ and $\Gamma$, a homomorphism $h:\Sigma\to\Gamma^*$ is said to be \emph{non-erasing} if $h(\sigma)\neq\lambda$ for all $\sigma\in\Sigma$, and $h$ is called  \emph{prefix-free} if there is no pair $\sigma,\tau\in\Sigma$ such that $h(\sigma)$ is a proper prefix of $h(\tau)$. As usual, $h$ is naturally expanded to a map from $\Sigma^*$ to $\Gamma^*$. The \emph{inverse} of $h$, denoted $h^{-1}$, is defined as $h^{-1}(y) = \{x\in\Sigma^*\mid h(x)=y\}$ for any $y\in\Gamma^*$.
It is important to note that the inverse $h^{-1}$ of a non-erasing prefix-free homomorphism $h$ becomes a partial function. Here, we wish to apply homomorphisms and inverse homomorphisms to partial functions.
We write $K$ for either the set of all homomorphisms or the set of all inverses of non-erasing prefix-free homomorphisms. We say that a class of families of partial functions is \emph{closed under} $K$ if, for any element $k$ in $K$ and for any $\FF=\{(f_n,D_n)\}_{n\in\nat}$ in this class, the family $\{(l_n,C_n)\}_{n\in\nat}$ also belongs to the class, where  $C_n=\{x\in\Sigma^*\mid k(x)\in D_n\}$ and $l_n(x) = f_n(k(x))$ for all $x\in C_n$.

\begin{lemma}\label{homomorphism}
$\onesharp$ is closed under both homomorphism and inverse of non-erasing prefix-free homomorphism.
\end{lemma}

\begin{proof}
We first examine the case of homomorphisms. Let us take any homomorphism $k:\Gamma\to\Sigma^*$ and expand it to a map from $\Gamma^*$ to $\Sigma^*$ in a standard way.
For any $\FF=\{(f_n,D_n)\}_{n\in\nat}$ over $\Gamma$, we then define $C_n=\{x\in\Sigma^*\mid k(x)\in D_n\}$ and $l_n(x) = f_n(k(x))$ for all numbers $n\in\nat$ and all strings $x\in C_n$. To show that $\{(l_n,C_n)\}_{n\in\nat}\in \onesharp$, it suffices to simulate an underlying 1nfa witnessing ``$\FF\in\onesharp$'' after executing $k$. Since $k$ does not depend on $n$, we can embed the behavior of $k$ into inner states of the desired simulator.

We next examine the case of inverses of non-erasing prefix-free homomorphisms.
Let us consider a family $\FF=\{(f_n,D_n)\}_{n\in\nat}$ of partial functions over alphabet, say, $\Sigma$ in $\onesharp$ and a non-erasing prefix-free homomorphism $h:\Sigma\to\Gamma^*$. As usual, we naturally expand $h$ to a map from $\Sigma^*$ to $\Gamma^*$. Since $h^{-1}$ becomes a function, we can  define $E_n=\{y\in\Gamma^* \mid h^{-1}(y)\in D_n\}$, which equals $\{h(x)\mid x\in D_n\}$. We further set $g_n(y) = f_n(h^{-1}(y))$ for all $y\in E_n$. For the obtained family  $\GG=\{(g_n,E_n)\}_{n\in\nat}$, we wish to verify that $\GG$ also belongs to $\onesharp$.

For the above purpose, let $\{M_n\}_{n\in\nat}$ denote a family of polynomial-size 1nfa's ``witnessing'' the membership of $\FF$ to $\onesharp$.  By modifying each 1nfa $M_n$, we introduce another 1nfa $N_n$ as follows. Let $H_{\Sigma}=\{h(\sigma)\mid \sigma\in\Sigma\}$. On input $y\in\Gamma^*$, $N_n$ moves its tape head over $y$ from left to right and identifies a prefix, say, $z$ of $y$ that appears in $H_{\Sigma}$.  Since $|z|$ is constantly bounded, it is possible to remember $z$ using $N_n$'s inner states.
Since $h$ is non-erasing and prefix-free, $z$ is not empty and there exists a unique symbol $\sigma\in\Sigma$ satisfying $h(\sigma)=z$.
The machine $N_n$ then runs $M_n$ on $\sigma$ to obtain its outcome.
We repeat this procedure on the rest of $y$ until $N_n$ completes the simulation of $M_n$'s computation on the entire string $h^{-1}(y)$. In the end, $N_n$ produces $g_n(y)$. This demonstrates that $\GG$ is indeed in $\onesharp$.
\end{proof}


We further seek out closure properties under various functional operations. Such properties were discussed by  Ogiwara and Hemachandra \cite{OH93} in the polynomial-time setting for $\sharpp$ and $\gapp$.

Let $\circ$ denote any binary functional operation. Given two partial functions $f$ and $g$ from $\Sigma^*$ to $\nat$ (or $\integer$) for a fixed alphabet $\Sigma$, we define new partial function $f\circ g$ by setting  $(f\circ g)(x) = f(x)\circ g(x)$ for all valid strings $x$ in $\Sigma^*$.
Let $\UU$ denote an arbitrary collection of families of partial functions. We say that $\UU$ is \emph{closed under} $\circ$ if, for any two families $\FF=\{(f_n,D_n)\}_{n\in\nat}$ and $\GG=\{(g_n,E_n)\}_{n\in\nat}$ in $\UU$, another family $\{(f_n\circ g_n,D_n\cap E_n)\}_{n\in\nat}$ also belongs to $\UU$.
As long as there is no confusion, we express the family $\{(f_n\circ g_n,D_n\cap E_n)\}_{n\in\nat}$ as $\FF\circ \GG$ for brevity.

\begin{lemma}\label{closure-onesharp}
The function classes $\onesharp$ and $\onegap$ are closed under addition and multiplication.
\end{lemma}

\begin{proof}
(1) We first focus on the case of $\onesharp$. Given any two function families $\FF=\{(f_n,D_n)\}_{n\in\nat}$ and $\GG=\{(g_n,E_n)\}_{n\in\nat}$ in $\onesharp$, we take two families $\{M_n\}_{n\in\nat}$ and $\{N_n\}_{n\in\nat}$ of polynomial-size 1nfa's that respectively ``witness'' the memberships of $\FF$ and $\GG$ to $\onesharp$. We intend to prove that $\FF\circ \GG$ is also in $\onesharp$ for any operator $\circ\in\{+,*\}$. For simplicity, we assume that $M_n$ and $N_n$ halt just after reading $\lhd$.

For ``addition'', we define a new 1nfa $H_n$ from $M_n$ and $N_n$ as follows. On input $x$, $H_n$ first branches out into two computation paths and then simulates $M_n$ on one path and $N_n$ on the other path. Clearly, $H_n$ produces $\#M_n(x)+\#N_n(x)$ accepting computation paths. Hence, if we set $h_n(x)=\#H_n(x)$ for all valid inputs $x$ and set $\FF+\GG=\{(h_n,D_n\cap E_n)\}_{n\in\nat}$, then $\FF+\GG$ belongs to $\onesharp$. We remark that, by the definition of partial functions, we do not need to check whether the input $x$ belongs to $D_n\cap E_n$ or not.

For ``multiplication'', we define another 1nfa $K_n$ in the following fashion. On input $x$, $K_n$ simulates $M_n$ and $N_n$ simultaneously by keeping their own inner states together as single inner states. This is possible because there is no $\lambda$-move conducted by $M_n$ and $N_n$. After both $M_n$ and $N_n$ halt simultaneously, if $M_n$ enters an accepting state, then we set $K_n$'s outcome (i.e., accepting or rejecting) to be $N_n$ 's outcome. Otherwise, $K_n$ ignores $N_n$'s outcome and simply enters a rejecting state.  This implies that the value $\#K_n(x)$ matches $\#M_n(x)\cdot \#N_n(x)$, which equals $f_n(x)\cdot g_n(x)$. We therefore define $k_n(x)=\#K_n(x)$ for all valid strings $x$. It then follows that $k_n=f_n\cdot g_n$, and thus $\FF\cdot \GG =\{(k_n,D_n\cap E_n)\}_{n\in\nat}$ belongs to $\onesharp$.

(2) Next, we turn our attention to the case of $\onegap$. Let $\FF=\{(f_n,D_n)\}_{n\in\nat}$ and $\GG=\{(g_n,E_n)\}_{n\in\nat}$ denote any two families of partial functions in $\onegap$. Without loss of generality, we assume that $D_n=E_n$ for all $n\in\nat$ since, otherwise, we can take $D_n\cap E_n$ as a new domain. By Lemma \ref{one-gap-character}(1), there are four families $\{(f_n^{(i)},D_n)\}_{n\in\nat}$ and $\{(g_n^{(i)},D_n)\}_{n\in\nat}$ for any $i\in\{1,2\}$ in $\onesharp$ satisfying that $f_n=f_n^{(1)}-f_n^{(2)}$ and $g_n=g_n^{(1)}-g_n^{(2)}$ for all $n\in\nat$. We then define $k_n= f_n^{(1)}\cdot g_n^{(1)}+f_n^{(2)}\cdot g_n^{(2)}$ and $h_n = f_n^{(1)}\cdot g_n^{(2)}+f_n^{(2)}\cdot g_n^{(1)}$. It follows that $f_n\cdot g_n = (f_n^{(1)}-f_n^{(2)}) (g_n^{(1)}-g_n^{(2)}) = (f_n^{(1)}\cdot g_n^{(1)}+f_n^{(2)}\cdot g_n^{(2)}) -  (f_n^{(1)}\cdot g_n^{(2)}+f_n^{(2)}\cdot g_n^{(1)}) = k_n - h_n$. By (1), $\{(k_n,D_n)\}_{n\in\nat}$ and $\{(h_n,D_n)\}_{n\in\nat}$ both fall in $\onesharp$.
Since $f_n\cdot g_n=k_n-h_n$, the family $\FF\cdot \GG =\{(f_n\cdot g_n,D_n)\}_{n\in\nat}$ belongs to $\onesharp - \onesharp$. Lemma \ref{one-gap-character}(1) thus concludes that $\FF\cdot\GG$ belongs to  $\onegap$.
The case of ``addition'' is similar.
Since $f_n+g_n = (f_n^{(1)}+g_n^{(1)})-(f_n^{(2)}+g_n^{(2)})$, it suffices to set $k'_n = f_n^{(1)}+g_n^{(1)}$ and $h'_n=f_n^{(2)}+g_n^{(2)}$. The family $\FF+\GG=\{(f_n+g_n,D_n)\}_{n\in\nat}$ coincides with $\{(k'_n-h'_n,D_n)\}_{n\in\nat}$, and thus it falls in $\onesharp-\onesharp$, which equals $\onegap$ by Lemma \ref{one-gap-character}(1).
\end{proof}

\subsection{Definitions of Counting Complexity Classes}\label{sec:def-complexity-class}

Let us formally introduce nonuniform polynomial state complexity classes,  associated with the notion of ``counting'', which will be used in the rest of this work.
In particular, we wish to formulate five counting complexity classes $\oneu$, $\oneparity$, $\onecequal$, $\onesp$, and $\onep$ (except for $\oned$ and $\onen$) using function families in $\onesharp$ and $\onegap$.

\begin{definition}\label{def:class-by-count}
(1) We first define the parity class $\oneparity$ (pronounced ``one parity'') as the collection of all families $\{(L_n^{(+)},L_n^{(-)})\}_{n\in\nat}$ of promise problems such that there exists a function family $\{(f_n,D_n)\}_{n\in\nat}$ in $\onesharp$ satisfying that,
for any index $n\in\nat$, (i) $L_n^{(+)}\cup L_n^{(-)}\subseteq D_n$ and (ii)  $f_n(x)$ is odd for all $x\in L_n^{(+)}$, and $f_n(x)$ is even for all $x\in L_n^{(-)}$.

(2) In a similar way, the unambiguous class $\oneu$ is obtained by replacing Condition (ii) in (1) with the following condition : $f_n(x)=1$ for all $x\in L_n^{(+)}$ and $f_n(x)=0$ for all $x\in L_n^{(-)}$.
\end{definition}

The aforementioned definition of $\onen$ given in Section \ref{sec:state-complexity} is also rephrased in terms of function families in $\onesharp$ by simply demanding the following condition: $f_n(x)>0$ for all $x\in L_n^{(+)}$
and $f_n(x)=0$ for all $x\in L_n^{(-)}$.

\begin{definition}\label{def:class-by-gap}
(1) The exact counting class $\onecequal$ (pronounced ``one C equal'') is defined as the collection of all families $\{(L_n^{(+)}, L_n^{(-)})\}_{n\in\nat}$ of promise
problems such that there exists a function family
$\{(g_n,D_n)\}_{n\in\nat}$ in $\onegap$ satisfying the following condition: for any index $n\in\nat$, (i) $L_n^{(+)}\cup L_n^{(-)}\subseteq D_n$ and (ii) $g_n(x)=0$ for all $x\in L_n^{(+)}$, and $g_n(x)\neq 0$ for all $x\in L_n^{(-)}$.

(2) In contrast, the stoic probabilistic class $\onesp$ refers to the collection of all families $\{(L_n^{(+)},L_n^{(-)})\}_{n\in\nat}$ defined by replacing the above condition (ii) with the following one: $g_n(x)=1$ for all $x\in L_n^{(+)}$, and  $g_n(x)=0$ for all $x\in L_n^{(-)}$.

(3) Finally, the bounded-error probabilistic class $\onep$ is defined with the following condition: $g_n(x)>0$ for all $x\in L_n^{(+)}$ and $g_n(x)\leq 0$ for all $x\in L_n^{(-)}$.
\end{definition}

We remark that $\onep$ is originally defined in \cite{Kap09,Kap12} in terms of unbounded-error probabilistic finite automata; as stated in Lemma  \ref{onecequalC-character}, nonetheless, it can be rephrased in terms of nonuniform families of polynomial-size 1nfa's
as in Definition \ref{def:class-by-gap}(3).
The complexity classes $\onecequal$ and $\onesp$ are closely related to  $\co\onen$ and $\oneu$ because the latter two classes are obtained directly by replacing $\onegap$ in the above definitions of $\onecequal$ and $\onesp$  with  $\onesharp$, respectively.

\section{Relationships among Counting Complexity Classes}\label{sec:relationship}

Throughout this section, we intend to discuss relationships among various counting complexity classes $\oneu$, $\onen$, $\onesp$, $\oneparity$, $\onecequal$, and $\onep$, introduced in Section \ref{sec:def-complexity-class}. In particular, we pay our attention to the containment/separation relationships of these classes.
Our results are illustrated in Figure \ref{fig:class-separations}.
In the log-space setting, however, such relationships (among $\mathrm{UL}$, $\nl$, $\mathrm{SPL}$, $\parityl$, $\cequall$, and $\pl$) are not yet known except for trivial ones.

\subsection{Basic Closure Properties}\label{sec:basic-closure-prop}

We start with a brief discussion on basic closure properties of $\oneparity$, $\onesp$, and $\onecequal$. Let us first take a close look at the closure property under \emph{complementation}. The complexity class $\oned$ is obviously closed under complementation. It is also proven in \cite{Yam22a} that $\onep=\co\onep$. As shown below, this closure property also holds for $\oneparity$ and $\onesp$.
This situation sharply contrasts the situation that $\oneu$ and $\onen$ are not closed under complementation \cite{Kap12,Yam22b}.

\begin{lemma}\label{parity-complement}
$\oneparity=\co\oneparity$ and $\onesp=\co\onesp$.
\end{lemma}

\begin{proof}
We first show that $\onesp=\co\onesp$. Take an arbitrary family $\LL=\{(L_n^{(+)},L_n^{(-)})\}_{n\in\nat}$ in $\onesp$ and choose a family $\FF=\{(f_n,D_n)\}_{n\in\nat}$ in $\onegap$ that ``witnesses'' the membership of  $\LL$ to $\onesp$. Moreover, we take a family $\{M_n\}_{n\in\nat}$ of polynomial-size 1nfa's that witnesses ``$\FF\in \onegap$''.
We modify each $M_n$ by swapping between the accepting states and the rejecting states of $M_n$ and by adding an extra accepting computation path  on every  input. The resulting  machine is denoted by $N_n$. It then follows that $\#N_n(x)=\#\overline{M}_n(x)+1$ and $\#\overline{N}_n(x)=\#M_n(x)$.
If we set $g_n(x)=\#N_n(x)-\#\overline{N}_n(x)$ for any $n\in\nat$ and $x\in D_n$, then the family $\GG=\{(g_n,D_n)\}_{n\in\nat}$ belongs to $\onegap$ and also makes $\co\LL$ fall in $\onesp$. Hence, this implies $\co\onesp\subseteq \onesp$, which is logically equivalent to $\onesp=\co\onesp$.

In a similar way, we obtain  $\oneparity = \co\oneparity$ by generating one extra accepting computation path on  each input given to an underlying 1nfa.
\end{proof}

It is known that $\onen$ is closed under intersection and union \cite{Yam22a} but not under complementation \cite{Kap12}. In comparison, we show the following closure properties of $\onecequal$ and $\co\onecequal$.

\begin{lemma}\label{cequal-closure}
$\onecequal$ is closed under intersection and $\co\onecequal$ is closed under union.
\end{lemma}

\begin{proof}
Let $\LL,\KK\in\onecequal$ have the form  $\LL=\{(L_n^{(+)},L_n^{(-)})\}_{n\in\nat}$ and $\KK=\{(K_n^{(+)},K_n^{(-)})\}_{n\in\nat}$ over the same alphabet $\Sigma$.  Take two families $\{M_n\}_{n\in\nat}$ and $\{N_n\}_{n\in\nat}$ of polynomial-size 1nfa's satisfying the following condition. For any instance $x\in\Sigma^*$, $x\in L_n^{(+)}$ implies $\#M_n(x)=\#\overline{M}_n(x)$, $x\in L_n^{(-)}$ implies $\#M_n(x) \neq \#\overline{M}_n(x)$,  $x\in K_n^{(+)}$ implies $\#N_n(x)=\#\overline{N}_n(x)$, and $x\in K_n^{(-)}$ implies $\#N_n(x) \neq \#\overline{N}_n(x)$. We denote by $(Q_{M,acc,n},Q_{M,rej,n})$ (resp., $(Q_{N,acc,n},Q_{N,rej,n})$) the pair of the set of accepting states and that of rejecting states of $M_n$ (resp., $N_n$).

Lemma \ref{cequal-path-number} makes it  possible to replace $\#M_n(x)\neq \#\overline{M}_n(x)$ and $\#N_n(x)\neq \#\overline{N}_n(x)$ given in the above conditions by $\#M_n(x) > \#\overline{M}_n(x)$ and $\#N_n(x) > \#\overline{N}_n(x)$, respectively.
For the intersection $\LL\cap\KK$, we construct another family $\PP=\{P_n\}_{n\in\nat}$ of 1nfa's as follows. On input $x$, $P_n$ runs both $M_n$ and $N_n$ simultaneously by keeping inner states $q_1$ of $M_n$ and $q_2$ of $N_n$ as pairs $(q_1,q_2)$. If both $M_n$ and $N_n$ accept (resp., reject), then $P_n$ accepts (resp., rejects). Otherwise, $P_n$ further branches out into one accepting state and one rejecting state.
For simplicity, let $m(x)$ denote $\#M_n(x)\cdot \#\overline{N}_n(x)+ \#\overline{M}_n(x)\cdot \#N_n(x)$. It then follows that $\#P_n(x)$ equals $\#M_n(x)\cdot \#N_n(x)+m(x)$ and that $\#\overline{P}_n(x)$ equals $\#\overline{M}_n(x)\cdot \#\overline{N}_n(x)+m(x)$.
If $x\in L_n^{(+)}\cap K_n^{(+)}$, then we obtain $\#P_n(x)=\#\overline{P}_n(x)$. In contrast, if $x\in L_n^{(-)}\cup K_n^{(-)}$, then $\#P_n(x)>\#\overline{P}_n(x)$ follows. These facts together conclude that $\PP$ solves $\LL\cap \KK$.

Clearly, the union closure of $\co\onecequal$ comes directly from the intersection closure of $\onecequal$.
\end{proof}


As for the closure property under complementation, similarly to those of $\onen$ and $\oneu$, we will see in Corollary \ref{Cequal-complement} that $\onecequal$ does not satisfy this property.

\subsection{Complexity Class 1U}\label{sec:oneu-class}

We next discuss containments and separations of $\oneu$ in comparison with other counting complexity classes.
Between $\oneu$ and $\onesp$, we can claim that $\oneu\subseteq \onesp$ and $\co\oneu\subseteq \onesp$. It is further possible to prove that these inclusions are in fact strict.
For later convenience, we write $T_n$ to denote the set $\{x\# y\mid x,y\in\{0,1\}^{2^n}\}$, which represents pairs of exponentially-long strings for each index $n\in\nat$.

\begin{proposition}\label{oneu-vs-oneparity}
$\oneu\subsetneqq \onesp$.
\end{proposition}

\begin{proof}
We first focus on the proof of $\oneu\subseteq \onesp$.
Let $\LL$ denote an arbitrary family $\{(L_n^{(+)},L_n^{(-)})\}_{n\in\nat}$ of promise problems in $\oneu$ and take a family $\MM=\{M_n\}_{n\in\nat}$ of polynomial-size 1nfa's solving $\LL$, where each $M_n$ is unambiguous (at least) on all valid inputs in $L_n^{(+)}\cup L_n^{(-)}$. From each $M_n$, we construct another 1nfa $N_n$ that behaves as follows.
On input $x$, start to simulate $M_n$. If $M_n$ enters an accepting state, then $N_n$ does the same. In contrast, whenever $M_n$ enters a rejecting state, $N_n$ branches out into two computation paths and it accepts on one path and rejects on the other. It then follows that $x\in L_n^{(+)}$ implies  $\#N_n(x)-\#\overline{N}_n(x)=1$ and $x\in L_n^{(-)}$ implies  $\#N_n(x)=\#\overline{N}_n(x)$. Hence, $\LL$ belongs to $\onesp$.

Next, we intend to show that $\oneu\neq \onesp$. We write $A$ to express the set $\{x\# y\mid x,y\in\{0,1\}^*\}$. The strings $x$ and $y$ in $x\# y$ in $A$ are respectively called the \emph{first block} and the \emph{second block} of $x\# y$.
For each index $n\in\nat$, we further define $L_n^{(+)}=\{x\# y\in A\mid \#_0(x)=\#_0(y)+1\}$ and $L_n^{(-)}=\{x\# y\in A\mid \#_0(x)=\#_0(y)\}$.
We then write $\LL_{sp}$ for the family $\{(L_n^{(+)},L_n^{(-)})\}_{n\in\nat}$ and claim that $\LL_{sp}$ belongs to $\onesp$. For this purpose, let us consider the following 1nfa $M_n$. On input $w$ of the form $x\# y$ in $A$, $M_n$ repeats the following procedure whenever reading a fresh input symbol until the right endmarker $\lhd$ is reached.  While reading the first block $x$ of $w$, if the currently accessed tape symbol $\sigma$ is $0$, then $M_n$ branches out into three computation paths.
On one path, $M_n$ immediately rejects, and on the two other paths, $M_n$ immediately accepts. When $\sigma$ is $1$, on the contrary, we just skip to the next input symbol.
The machine $M_n$ also skips $\#$ and continues the above procedure by swapping between accepting and rejecting states. On reading $\lhd$, $M_n$ branches out into two transitions with one accepting state and one rejecting state.
It then follows that $\#_0(x) = \#M_n(x) - \#\overline{M}_n(x)$ and $\#_0(y) = \#M_n(y)- \#\overline{M}_n(y)$. Thus, we obtain $\#_0(x)-\#_0(y) = \#M_n(x)+\#M_n(y) - ( \#\overline{M}_n(x) + \#\overline{M}_n(y) ) = \#M_n(x\# y) - \#\overline{M}_n(x\# y)$.
The last equality clearly implies that $\LL_{sp}\in \onesp$.

Finally, we claim that $\LL_{sp}\notin \oneu$. Assuming otherwise, we take a family $\NN=\{N_n\}_{n\in\nat}$ of polynomial-size 1nfa's solving $\LL_{sp}$ such that, for any $n\in\nat$, $st(N_n)\leq n^k$ holds for a certain absolute constant $k>0$ and $N_n$ is unambiguous (at least) on $L^{(+)}_n\cup L^{(-)}_n$. Let $Q_n$ denote the set of all inner states of $N_n$.
For every string $x\# y$ in $L_n^{(+)}$, the notation $\mu(x,y)$ denotes a unique inner state $q$ satisfying that there exists a unique accepting computation path $\gamma$ along which $N_n$ enters $q$ just after reading off $x\#$.
Fix an arbitrary $n\in\nat$ satisfying $n^k<2^{n/2}$. We set $I_n$ to be $\{x\# y\in L_n^{(+)}\mid \#_0(x)\geq 2^{n/2}\}$. We choose two strings $x_1\# y_1$ and $x_2\# y_2$ in $I_n$ for which $\#_0(x_1) \neq \#_0(x_2)$ and $\mu(x_1,y_1) = \mu(x_2,y_2)$. This is possible because the total number of different values of $\mu(x,y)$ over all strings $x\# y$ in $I_n$ is upper-bounded by $|Q_n|$.

Let us consider the following computation on the input $x_1\# y_2$. For simplicity, we write $q$ for the inner state $\mu(x_1,y_1)$. We first simulate $N_n$ on $x_1\#$ and then enter $q$. Starting with $q$, we simulate $N_n$ on $y_2$ and then enter our own accepting state whenever entering an accepting state of $N_n$. Note that this procedure forms a legitimate accepting computation path of $N_n$ on $x_1\# y_2$. Thus, $x_1\# y_2\in L_n^{(+)}$ follows. This is a clear contradiction to the definition of $L_n^{(+)}$.
\end{proof}


In what follows, we intend to prove the separation of $\co\oneu \nsubseteq \onen$.
In its proof, we use the notion of Kolmogorov complexity.
Given a binary string $x$, the notation $C(x)$ denotes the \emph{(unconditional) Kolmogorov complexity} of $x$; namely, the minimum length of a binary-encoded program $p$ such that a \emph{universal Turing machine} $U$ takes $p$ as an  input and produces $x$ on its output tape in a finite number of steps. See, e.g., \cite{LV97} for more detailed information.


Given numbers $k,m,n\in\nat^{+}$ and $i_1,i_2,\ldots,i_k\in[n]$ with $n\leq m$, we introduce the notation $\dbraleft i_1,i_2,\ldots,i_k \dbraright_m$ to represent the binary string $1^{i_1}0^{m-i_1}0 1^{i_2}0^{m-i_2}0 \cdots 0 1^{i_k}0^{m-i_k}$, where we treat both $1^0$ and $0^0$ as $\lambda$. Let $B_n(m,k)$ denote the set of all strings of the form $\dbraleft i_1,i_2,\ldots,i_k \dbraright_m$ with  $i_1,i_2,\ldots,i_k\in[0,n]_{\integer}$, provided that $n\leq m$.
We further write $B_n(m)$ for the infinite union $\bigcup_{k\in\nat^{+}} B_n(m,k)$.
It follows that $|\dbraleft i_1,i_2,\ldots,i_k \dbraright_m| = km+k-1$,  whereas $|[i_1,i_2,\ldots,i_k]| = \sum_{l=1}^{k}i_l+k$. See \cite{Yam22b} for more information.

\begin{theorem}\label{separate-co-oneu}
$\co\oneu\nsubseteq \onen$
\end{theorem}

\begin{proof}
We first define an example family $\LL_U=\{(L_n^{(+)},L_n^{(-)})\}_{n\in\nat}$
of promise problems as follows.
For every index $n\in\nat$, we set $L_n^{(+)} = \{u\# v\mid u,v\in B_n(n,n), \exists!e\in[n]( (u)_{(e)}\neq (v)_{(e)} )\}$ and $L_n^{(-)}=\{u\# v\mid u,v\in B_n(n,n), \forall e\in[n]( (u)_{(e)} = (v)_{(e)} )\}$. Notice that $|u\# v|=2n^2+2n-1$ if $u\# v\in L_n^{(+)}\cup L_n^{(-)}$, since $u\in B_n(n,n)$ implies $|u|=n^2+n-1$.

It then follows that $\LL_{U}$ belongs to $\oneu$.
To see this fact, let us consider the following 1nfa $M_n$.
On input $x$ of the form $u\# v$, check if $u$ has the form $\dbraleft i_1,i_2,\ldots,i_n \dbraright_n$.
At the same time, nondeterministically choose an index $e\in[0]$, read $(u)_{(e)}$, and remember both $e$ and $(u)_{(e)}$.
Similarly, after moving to $v$, check if $v$ has the form $\dbraleft i'_1,i'_2,\ldots,i'_n \dbraright_n$ and read $(v)_{(e)}$. Finally, check if $(u)_{(e)}\neq (v)_{(e)}$. When $x\in L^{(+)}_n$, $M_n$ produces exactly one accepting computation path. When $x\in L^{(-)}_n$, on the contrary, all computation paths of $M_n$ are rejecting ones. Notice that there may be a chance of $M_n$'s producing more than one accepting paths on a certain ``invalid'' input.

Next, we focus on the complement of $\LL_{U}$, that is,  $\co\LL_{U}=\{(L_n^{(-)},L_n^{(+)})\}_{n\in\nat}$, which belongs to   $\co\oneu$.
In what follows, we intend to verify that $\co\LL_{U}\notin \onen$.
To lead to a contradiction, we assume the existence of a family $\MM=\{M_n\}_{n\in\nat}$ of polynomial-size 1nfa's that solves $\co\LL_{U}$. For each index $n\in\nat$, let $Q_n$ denote the set of all inner states of $M_n$.
Note that there exists a polynomial $p$ satisfying
$|Q_n|\leq p(n)$ for all $n\in\nat$.  It is possible to assume that all elements of $Q_n$ are expressed as numbers in $[0,|Q_n|-1]_{\integer}$.
Given any $n\in\nat$ and any string $u\in B_n(n,n)$, we define $S_n(u)$ to be the set of all inner states $q$ in $Q_n$ such that, in a certain accepting computation path of $M_n$ on the input $u\# u$, $M_n$ enters $q$ just after reading off $u\#$.

Take a sufficiently large integer $n$ satisfying $p(n)<2^n$. Let us consider the Kolmogorov complexity $C(w)$ of string $w$. We choose a string $u$ in $B_n(n,n)$ for which $C(u)\geq n\log{n}$.
Such a string exists  because, otherwise, $|B_n(n,n)|\leq |\{u\in\{0,1\}^{n^2+n-1} \mid C(u) < n\log{n}\}| < 2^{n\log{n}}$, a contradiction to $|B_n(n,n)|=(n+1)^n$.
We then choose an inner state $q\in S_n(u)$, which is the smallest value in $[0,|Q_n|-1]_{\integer}$.
Note that the length of the binary expression of $q$ is $O(\log{n})$ since $|Q_n|\leq p(n)$.

Let us design a deterministic algorithm, say, $\BB$ that behaves as follows. By choosing all strings $w$ in $B_n(n,n)$ inductively one by one, $\BB$ runs $M_n$ on the input $w\# w$ starting with the inner state $q$. If $M_n$ enters no accepting states on all computation paths, then $\BB$ outputs $w$ and halts. Otherwise, $\BB$ chooses another $w$ and continues the above procedure.
Let $w'$ denote the outcome of this algorithm $\BB$. The string $w'$ must be $u$ because, otherwise, $M_n$ also accepts $u\# w'$, a contradiction. Let $r_0$ denote the binary encoding of the  algorithm $\BB$. It is important to note that $|r_0|$ is a constant independent of $u$ and $w'$. It thus follows that $C(u)\leq |bin(n)|+|bin(q)|+ |r_0| +O(1) \leq O(\log{n})$, which is in contradiction to the inequality $C(u)\geq n\log{n}-1$.
\end{proof}


Since $\oneu\subseteq \onen$, Theorem  \ref{separate-co-oneu} instantly  yields the known separation of $\oneu\neq\co\oneu$ \cite{Yam22b,Yam22c}. This theorem also leads to the following important consequence regarding the complexity of $\onesp$.

\begin{corollary}\label{1N-vs-1SP}
$\onesp\nsubseteq \onen$.
\end{corollary}

\begin{proof}
Assume that $\onesp\subseteq \onen$. Since $\oneu \subseteq \onesp$ by Proposition \ref{oneu-vs-oneparity}, it follows that $\co\oneu\subseteq \co\onesp = \onesp$. We thus obtain $\co\oneu\subseteq \onen$. This obviously contradicts Theorem \ref{separate-co-oneu}. Therefore, we can conclude that $\onesp\nsubseteq \onen$.
\end{proof}

\subsection{Complexity Class 1C$_{=}$}\label{sec:class-onecequal}

Let us discuss containment and separation relationships concerning  $\onecequal$.
A core of this subsection is an exploitation of a close connection between nonuniform (polynomial) state complexity classes of this work and one-tape linear-time Turing machines with linear-size advice discussed in \cite{TYL10,Yam10}.
This close connection helps us adapt a key lemma (Lemma 4.3) of \cite{Yam10} in our setting and exploit it to prove the class separations between $\onen$ and $\onecequal$, as shown below.

\begin{theorem}\label{N-vs-cequal}
\renewcommand{\labelitemi}{$\circ$}
(1) $\onen\subsetneqq \co\onecequal$ and $\co\onen\subsetneqq
\onecequal$.
(2) $\onen\nsubseteq\onecequal$.
\end{theorem}

Recall from Section \ref{sec:basic-closure-prop} the closure property of $\onecequal$ under intersection. By sharp contrast, we prove the following non-closure property of $\onecequal$ as an immediate consequence of Theorem \ref{N-vs-cequal}.

\begin{corollary}\label{Cequal-complement}
$\onecequal$ is not closed under complementation.
\end{corollary}

\begin{proof}
Theorem \ref{N-vs-cequal}(1) implies $\onen\subseteq \co\onecequal$.  Moreover, if $\onecequal=\co\onecequal$, then $\onen\subseteq \onecequal$ follows. However, Theorem \ref{N-vs-cequal}(2) shows that $\onen\nsubseteq \onecequal$. The corollary thus follows immediately.
\end{proof}

This non-closure property of $\onecequal$ is helpful to prove the following statement.

\begin{proposition}\label{onesp-vs-onecequal}
$\onesp\subsetneqq \co\onecequal$ and $\onesp\subsetneqq \onecequal$.
\end{proposition}

\begin{proof}
The inclusion $\onesp\subseteq \co\onecequal$ is trivial. The separation  $\onesp\neq \co\onecequal$ follows directly from Corollary \ref{Cequal-complement} because, otherwise, we obtain  $\onecequal=\co\onecequal$ by Lemma \ref{parity-complement}, leading to a contradiction.  The latter claim of $\onesp\subsetneqq\onecequal$ in the proposition follows from $\onesp\subsetneqq \co\onecequal$ and $\onesp=\co\onesp$.
\end{proof}


As another consequence of Theorem \ref{N-vs-cequal}, we obtain the following.

\begin{proposition}\label{onecequal-vs-onep}
$\onecequal \subsetneqq \onep$.
\end{proposition}

\begin{proof}
Let $\LL$ denote any family in $\onecequal$ and consider a family $\MM=\{M_n\}_{n\in\nat}$ of polynomial-size 1nfa's for which $\MM$ ``witnesses'' the membership of $\LL$ to $\onecequal$.
Lemma \ref{cequal-path-number} provides another family $\{N_n\}_{n\in\nat}$ of polynomial-size 1nfa's that satisfies the lemma's conditions, namely, for any string $x$, (i) $\#M_n(x)=\#\overline{M}_n(x)$ implies $\#N_n(x)=\#\overline{N}_n(x)$ and (ii) $\#M_n(x)\neq \#\overline{M}_n(x)$ implies  $\#N_n(x)>\#\overline{N}_n(x)$. Therefore, $\LL$ belongs to $\co\onep$.
Since $\onep=\co\onep$  \cite{Yam22a}, this implies that $\onecequal\subseteq \onep$.

Next, we intend to prove that $\onecequal \neq \onep$. Notice that this inequality implies $\co\onecequal\nsubseteq \onep$ as well because of  $\onep=\co\onep$. Assume that $\onecequal=\onep$. Since $\onen\subseteq \onep$ \cite{Yam22a}, we obtain $\onen\subseteq \onecequal$. This clearly contradicts Theorem  \ref{N-vs-cequal}(2).
\end{proof}


In the rest of this subsection, we intend to prove Theorem \ref{N-vs-cequal}. Let us recall $1\mbox{-}\mathrm{C}_{=}\mathrm{LIN/lin}$, which is composed of all languages recognized by one-tape linear-time exact counting TMs equipped with linear-size advice; that is, the following acceptance criteria hold: each TM accepts exactly when there are equal numbers of accepting and rejecting computation paths on all inputs.
The desired proof requires an idea from \cite{Yam10}, more specifically,
a useful, characteristic property of $1\mbox{-}\mathrm{C}_{=}\mathrm{LIN/lin}$, presented in \cite[Lemma 4.3]{Yam10}, to demonstrate that $1\mbox{-}\mathrm{C}_{=}\mathrm{LIN/lin}$ is not closed under complementation.
As noted in Section \ref{sec:exploitation}, there is a close connection
between one-tape linear-time TMs and finite automata.
This connection seems to be adapted to one-tape linear-time ``advised'' TMs and ``nonuniform'' finite automata families. In particular, we look into a connection
between $1\mbox{-}\mathrm{C}_{=}\mathrm{LIN/lin}$ and $\onecequal$. We  wish to exploit this connection to translate the above property of $\mathrm{1\mbox{-}C_{=}LIN/lin}$ into the setting of $\onecequal$ and to
achieve the desired separation results of Theorem \ref{N-vs-cequal}.

\begin{lemma}\label{property-1cequal}
Let $\LL = \{(L_n^{(+)},L_n^{(-)})\}_{n\in\nat}$ denote any family in $\onecequal$ over alphabet $\Sigma$. There exists a polynomial $p$ that satisfies the following statement. Let $n$, $m$, and $l$ be any numbers in $\nat$ with $l\leq m-1$, let $z\in\Sigma^{l}$, and let $A_{n,m,l,z}= \{x\in\Sigma^{m-l}\mid xz\in L_n^{(+)}\}$. There exists a subset $S$ of $A_{n,m,l,z}$ with $|S|\leq p(n)$ such that, for any $y\in\Sigma^l$, if $\{wy\mid w\in S\}\subseteq L_n^{(+)}$, then $\{xy\mid x\in A_{n,m,l,z}\}\subseteq L_n^{(+)}$.
\end{lemma}

Meanwhile, we postpone the proof of Lemma \ref{property-1cequal} and we wish  to prove Theorem \ref{N-vs-cequal} with the use of the lemma.

\vs{-2}
\begin{proofof}{Theorem \ref{N-vs-cequal}}
(1) We attempt to verify that $\onen\subsetneqq \co\onecequal$.
We first claim that $\onen\subseteq \co\onecequal$. For any family $\LL=\{(L_n^{(+)},L_n^{(-)})\}_{n\in\nat}$ of promise problems in $\onen$, let us consider a family $\FF=\{(f_n,D_n)\}_{n\in\nat}$ of partial functions in $\onesharp$ such that (*)  $f_n(x)>0$ for all $x\in L_n^{(+)}$ and $f_n(x)=0$ for all $x\in L_n^{(-)}$.
Since $\onesharp\subseteq \onegap_{\geq0}$ by Lemma \ref{inclusion-1F-to-1sharp}(1), $\FF$ belongs to $\onegap_{\geq0}$. Condition (*) then implies that $\LL$ belongs to $\co\onecequal$ by the definition of $\onecequal$. Therefore, we conclude that $\onen\subseteq \co\onecequal$.

The remaining task is to prove $\onen\neq \co\onecequal$. Since $\onesp\subseteq \co\onecequal$, $\onen = \co\onecequal$ implies $\onesp\subseteq \onen$, which further implies $\co\oneu\subseteq \onen$ by lemma \ref{parity-complement} and Proposition \ref{oneu-vs-oneparity}.
Since $\co\oneu\nsubseteq\onen$ by Theorem \ref{separate-co-oneu},
we obtain the desired separation $\onen\neq \co\onecequal$.

(2) Here, we intend to verify that $\onen\nsubseteq \onecequal$.
Our example family $\LL_N=\{(L_n^{(+)},L_n^{(-)})\}_{n\in\nat}$ is introduced  as follows. Let $\Sigma=\{0,1\}$. Assume that $n\geq1$. For simplicity, we write $I_n$ for the set $B_{n^2}(n^2,n)$ ($=\{\dbraleft i_1,i_2,\ldots,i_n\dbraright_{n^2} \mid i_1,i_2,\ldots,i_n\in[0,n^2]_{\integer}\}$). Notice that  $|z|=n^3+n-1$ holds for any string $z\in I_n$. The desired family $\LL_{N}$ is obtained by setting $L_n^{(+)} = \{u\# v\mid u,v\in I_n, Set(u)\neq Set(v)\}$ and $L_n^{(-)} = \{u\# v\mid u,v\in I_n, Set(u)=Set(v)\}$ for each index $n\in\nat^+$. It is important to remark that, for any $u\in I_n$, $u\# u$ always belongs to $L_n^{(-)}$.
It is not difficult to show that $\LL_N$ belongs to $\onen$.

Hereafter, we wish to prove by contradiction that $\LL_N\notin \onecequal$. Assuming that $\LL_{N}\in \onecequal$, we apply Lemma \ref{property-1cequal} to $\LL_{N}$ and take a polynomial $p$ that satisfies the lemma. Take a sufficiently large number $n$. We then set $m=2n^3+2n-1$ and $l=n^3-3n-1$. We choose the specific string $z=\dbraleft 1,2,\ldots,n\dbraright_{n^2}$ in $I_n$ and consider the set $A_{n,m,l,z}$, which equals $\{x\#\in\Sigma^{m-l}\mid x\#z\in L_n^{(+)}\}$.
There exists a subset $S$ of $A_{n,m,l,z}$ with $|S|\leq p(n)$ satisfying the lemma. For convenience, we introduce the notation $P_n$ to denote the set $\{\dbraleft i_1,i_2,\ldots,i_n \dbraright_{n^2}\in I_n \mid \text{ $(i_1,i_2,\ldots,i_n)$ is a permutation of $(1,2,\ldots,n)$ }\}$.
Choose a string $y$ in $I_n- \{u\mid u\in P_n\text{ or } \exists w [w\#\in S \wedge Set(w)=Set(u)]\}$. Since $y\# z\in L_n^{(+)}$, the string $y\#$ belongs to $A_{n,m,l,z}$. It also follows by the definition of $y$ that $\{w\#y\mid w\#\in S\}\subseteq L_n^{(+)}$.  The lemma then concludes that $\{x\#y\mid x\#\in A_{n,m,l,z}\}\subseteq L_n^{(+)}$.
Since $y\#\in A_{n,m,l,z}$, we obtain $y\# y\in L_n^{(+)}$ in  contradiction to the definition of $\LL_N$.
In conclusion, $\LL_N\notin \onecequal$ follows.
\end{proofof}


For the proof of Lemma \ref{property-1cequal}, it is useful to recall from \cite{Yam22a} another matrix way of defining $\onep$ using \emph{one-way probabilistic finite automata} (or 1pfa's, for short) instead of using 1nfa's (as given in Definition \ref{def:class-by-gap}(3)).
A 1pfa $N$ is defined as a sextuple $(Q,\Sigma, \nu_{ini}, \{M_{\sigma}\}_{\sigma\in\check{\Sigma}}, Q_{acc}, Q_{rej})$, where $\nu_{ini}$ is an initial state (row) vector of $|Q|$ dimension with rational entries, each $M_{\sigma}$ is a $|Q|\times|Q|$ \emph{stochastic matrix}\footnote{A nonnegative square matrix $M$ is \emph{stochastic} if every row of $M$ sums up to exactly $1$. In the past literature, columns of a matrix are used for the definition of stochasticity.} with rational entries, and $Q_{acc}$ and $Q_{rej}$ ($\subseteq Q$) are respectively the set of accepting inner states and that of rejecting inner states. For each index $h\in\{acc,rej\}$, we prepare a row vector $\xi_{Q_{h}}$ whose $q$-entry is $1$ if $q\in Q_h$ and $0$ otherwise. Given a length-$n$ string  $x=x_1x_2\cdots x_n$ in $\check{\Sigma}^*$, we define $M_x$ as $M_{x_1}M_{x_2}\cdots M_{x_n}$. The acceptance (resp., rejection) probability of $N$ on an input $x\in\Sigma^*$ is set to be $p_{acc}(x)=\nu_{ini}M_{\triangleright x\triangleleft}\xi_{Q_{acc}}^T$ (resp., $p_{rej}(x)=\nu_{ini}M_{\triangleright x\triangleleft}\xi_{Q_{rej}}^T$), where the superscript $T$ means ``transpose''.


In \cite{Yam22a}, $\onep$ is introduced using families of polynomial-size 1pfa's with unbounded-error probability as follows.
We say that $\{M^{(n)}_{\sigma}\}_{\sigma\in\check{\Sigma},n\in\nat}$ is \emph{polynomially manageable} if there is a polynomial $p$ such that any entries $r$ of each matrix $M^{(n)}_{\sigma}$ are lower-bounded by $1/p(n)$.

\begin{lemma}{\rm \cite{Yam22a}}\label{onecequalC-character}
Given a family $\LL=\{(L_n^{(+)},L_n^{(-)})\}_{n\in\nat}$ of promise problems, $\LL$ is in $\onep$ iff there exists a family $\{M_n\}_{n\in\nat}$ of polynomial-size 1pfa's $M_n$ of the form
$(Q_n,\Sigma,\nu_{ini}^{(n)}, \{M_{\sigma}^{(n)}\}_{\sigma\in\check{\Sigma}},  Q_{acc,n},Q_{rej,n})$
satisfying that (1) $\{M^{(n)}_{\sigma}\}_{\sigma\in\check{\Sigma},n\in\nat}$ is polynomially manageable and (2) for any $n\in\nat$, $p_{acc,n}(x)>1/2$ holds for all $x\in L_n^{(+)}$ and $p_{rej,n}(x)\geq 1/2$
holds for all $x\in L_n^{(-)}$.
\end{lemma}


For $\onecequal$, we can give a similar characterization using 1pfa's.

\begin{lemma}\label{exact-counting-1pfa}
Given a family $\LL=\{(L_n^{(+)},L_n^{(-)})\}_{n\in\nat}$ of promise problems, $\LL$ is in $\onecequal$ iff there exists a family $\{M_n\}_{n\in\nat}$ of polynomial-size 1pfa's $M_n$ of the form $(Q_n,\Sigma,\nu_{ini}^{(n)}, \{M_{\sigma}^{(n)}\}_{\sigma\in\check{\Sigma}},  Q_{acc,n},Q_{rej,n})$ such that (1) $\{M^{(n)}_{\sigma}\}_{\sigma\in\check{\Sigma},n\in\nat}$ is polynomially manageable and (2) for any $n\in\nat$, $p_{acc,n}(x)=1/2$ holds for all strings $x\in L_n^{(+)}$ and $p_{acc,n}(x)\neq 1/2$ holds for all $x\in L_n^{(-)}$, where $p_{acc,n}(x)$ indicates the acceptance probability of $M_n$ on input $x$.
\end{lemma}

\begin{proofsketch}
Any 1nfa can be treated as a 1pfa whose computation paths are produced with equal probability. This shows the ``only if'' direction.
For the ``if'' direction, we take a family $\{M_n\}_{n\in\nat}$ of polynomial-size 1pfa's satisfying all the conditions of the lemma. When $M_n$ in inner state $p$ produces $k_p$ branches with inner states $q_1,q_2,\ldots,q_{k_p}$ with probabilities $\frac{m_1}{e_{n,p}},  \frac{m_2}{e_{n,p}}, \ldots, \frac{m_{k_p}}{e_{n,p}}$, respectively, we can produce $m_1e_{n,p},m_2e_{n,p},\ldots,m_{k_p}e_{n,p}$ branches with the same inner states $q_1,q_2,\ldots,q_{k_p}$, respectively, where $e_{n,p}$ is a polynomially-bounded positive integer. Notice that $\sum_{i=1}^{k_p}m_i = e_{n,p}$. It follows that each computation path $\gamma=(p_1,p_2,\ldots,p_{|\rhd{x}\lhd|})$ of $M_n$ generated with probability $\alpha$ is simulated by $\alpha\prod_{i=1}^{|\rhd{x}\lhd|} e_{n,p_i}$ computation paths of $N_n$. This implies that $\#N_n(x)+\#\overline{N}_n(x) =\sum_{\gamma} (\prod_{i=1}^{|\rhd{x}\lhd|}e_{n,p_i})$, where $\gamma=(p_1,p_2,\ldots,p_{|\rhd{x}\lhd|})$ ranges over all possible computation paths of $M_n$. Moreover, it follows that $p_{acc,n}=1/2$ iff $\#N_n(x)=1/2(\#N_n(x)+\#\overline{N}_n(x))$.
\end{proofsketch}

Now, we return to the proof of Lemma \ref{property-1cequal}.

\vs{-2}
\begin{proofof}{Lemma \ref{property-1cequal}}
Let $\LL=\{(L_n^{(+)},L_n^{(-)})\}_{n\in\nat}$ denote any family of promise problems in $\onecequal$.
By Lemma \ref{exact-counting-1pfa}, there exists a family $\NN=\{N_n\}_{n\in\nat}$ of polynomial-size 1pfa's solving $\LL$ with the following conditions: for any index $n\in\nat$,  $p_{acc,n}(x)=1/2$ holds for all $x\in L_n^{(+)}$ and holds $p_{acc,n}(x)\neq 1/2$ for all $x\in L_n^{(-)}$.
The polynomial state complexity of $\NN$ helps us choose a suitable polynomial $p$ for which $N_n$'s inner-state set $Q_n$ has at most $p(n)$ elements for any index $n\in\nat$.
In the following argument, we focus our attention only on valid strings in $L_n^{(+)}\cup L_n^{(-)}$.

Hereafter, we fix $n$ arbitrarily and
assume that $N_n$ has the form $(Q_n,\Sigma, \nu_{ini}^{(n)}, \{M_{\sigma}^{(n)}\}_{\sigma\in\check{\Sigma}}, Q_{acc,n}, Q_{rej,n})$ with the (row) vectors $\xi_{Q_{acc,n}}$ and $\xi_{Q_{rej,n}}$ induced from $Q_{acc,n}$ and $Q_{rej,n}$, respectively.
We remark that  $p_{acc,n}(x)=\nu_{ini}^{(n)} M_{\triangleright x \triangleleft}^{(n)} \xi_{Q_{acc,n}}^T$ and  $p_{rej,n}(x)=\nu_{ini}^{(n)} M_{\triangleright x \triangleleft}^{(n)} \xi_{Q_{rej,n}}^T$ hold for any $x\in L_n^{(+)}\cup L_n^{(-)}$.

Let us take two integers $m,l\in\nat$ with $l<m-1$ and a string $z\in\Sigma^l$ arbitrarily.
The lemma is trivial when $l=0$. Therefore, we assume that $l>0$. For the set $A_{n,m,l,z}$ in the premise of the lemma, if $|A_{n,m,l,z}|\leq |Q_n|$, then it suffices to define $S$ to be  $A_{n,m,l,z}$ because of $|Q_n|\leq p(n)$. In what follows, we examine the case of $|A_{n,m,l,z}|>|Q_n|$.

Let us consider the set $V_n=\{\nu_{ini}^{(n)} M_{\triangleleft w}^{(n)} \mid w\in A_{n,m,l,z}\}$ induced from $A_{n,m,l,z}$.
Since $V_n$ consists only of $|Q_n|$-dimensional vectors, it is possible to  select a maximal subset $S'$ of linearly independent vectors in $V_n$.
Note that $|S'|\leq |Q_n|\leq p(n)$. From the set $S'$, we define $S$ to be the set $\{w\in A_{n,m,l,z}\mid \nu_{ini}^{(n)} M_{\triangleright w}^{(n)} \in S'\}$. It then follows that, for any valid string $x\in L^{(+)}_n\cup L^{(-)}_n$, there always exists a series $\{\alpha^{(x)}_w\}_{w\in S}$ of nonnegative real numbers for which  $\nu_{ini}^{(n)} M_{\triangleright x}^{(n)} = \sum_{w\in S} \alpha^{(x)}_w \nu_{ini}^{(n)} M_{\triangleright w}^{(n)}$ holds with
$\sum_{w\in S} \alpha^{(x)}_w =1$.
Let $y$ denote any string in $(L_n^{(+)}\cup L_n^{(-)})\cap \Sigma^l$ for which
$\{wy\mid w\in S\} \subseteq L_n^{(+)}$. Notice that, for any $w\in S$, since $wy\in L_n^{(+)}$, we obtain $p_{acc,n}(wy)=1/2$. For every $x\in A_{n,m,l,z}$, it thus follows that  $p_{acc,n}(xy)= \nu_{ini}^{(n)} M_{\triangleright x}^{(n)} M_{y\triangleleft}^{(n)} \xi_{Q_{acc,n}}^T = (\sum_{w\in S}\nu_{ini}^{(n)} M_{\triangleright w}^{(n)}) M_{y\triangleleft}^{(n)} \xi_{Q_{acc,n}}^T = \sum_{w\in S} \alpha^{(x)}_w p_{acc,n}(wy) = \frac{1}{2}\sum_{w\in S}\alpha^{(x)}_{w} =  1/2$. This implies that $xy$ belongs to $L_n^{(+)}$.

This completes the proof.
\end{proofof}

\subsection{Complexity Class 1P}\label{sec:onep-class}

We turn our attention to $\onep$, whose elements are characterized by gap function families in $\onegap$. In the polynomial-time setting, $\pp$ is known to be as hard as the polynomial hierarchy \cite{Tod91}. A key statement of this subsection is the separation between $\oneparity$ and $\onep$.

\begin{theorem}\label{parity-vs-onep}
$\oneparity\nsubseteq \onep$.
\end{theorem}

Meanwhile, we postpone the description of the proof of this theorem.
The theorem makes us
obtain another class separation between $\onesp$ and $\oneparity$.

\begin{proposition}\label{onesp-vs-oneparity}
$\onesp\subsetneqq \oneparity$.
\end{proposition}

\begin{proof}
We first claim that $\onesp\subseteq \oneparity$. Let $\LL=\{(L_n^{(+)},L_n^{(-)})\}_{n\in\nat}$ denote any family in $\onesp$ and choose a family $\{(f_n,D_n)\}_{n\in\nat}$ in $\onegap$ that ``witnesses'' the membership of $\LL$ to $\onesp$.
Let $\{M_n\}_{n\in\nat}$ denote a family of polynomial-size 1nfa's satisfying $f_n(x)=\#M_n(x)-\#\overline{M}_n(x)$ for all $n\in\nat$ and all $x\in D_n$. We then modify $M_n$ into $N_n$ as follows.
The machine $N_n$ first runs $M_n$ on input $x$. On each computation path  of $M_n$, if $M_n$ reaches an accepting state, then $N_n$ does the same; otherwise, $N_n$ branches out into two extra computation paths and $N_n$ accepts on one path and rejects on the other path.
If $x\in L_n^{(+)}$, then we obtain $M_n(x)-\overline{M}_n(x)=1$ because of $f_n(x)=1$, and thus $\overline{M}_n(x)=M_n(x)-1$ follows. Hence, $\#N_n(x)$ equals $\#M_n(x)+\#\overline{M}_n(x)=2\#M_n(x)-1$, which is an odd number.
In contrast, if $x\in L_n^{(-)}$, then $\#N_n(x)$ equals $\#M_n(x)+\#\overline{M}_n(x)=2\#M_n(x)$, which is an even number, because of  $f_n(x)=0$.
Finally, we define $g_n$ by setting $g_n(x)=\#N_n(x)$ for all $x$. Since $\{(g_n,D_n)\}_{n\in\nat}$ belongs to $\onesharp$, we conclude that $\LL$ is in $\oneparity$.

Next, we assert that $\oneparity\neq \onesp$.  Assume otherwise. We then obtain $\oneparity\subseteq \onecequal$ since $\onesp\subseteq \onecequal$ by Proposition \ref{onesp-vs-onecequal}. As a consequence, $\oneparity\subseteq \onep$ follows. However, this contradicts Theorem \ref{parity-vs-onep}. Therefore, we conclude that $\oneparity\neq \onesp$.
\end{proof}

To prove Theorem \ref{parity-vs-onep}, nonetheless, we need to resort a close connection between one-tape linear-time TMs and finite automata and adapt it in our setting, in particular, between $1\mbox{-}\mathrm{PLIN/lin}$ and $\onep$, where $1\mbox{-}\mathrm{PLIN/lin}$ consists of languages recognized by one-tape linear-time unbounded-error probabilistic TMs with linear-size advice.
In a demonstration of the power of $\cfl$ over $1\mbox{-}\mathrm{PLIN/lin}$, a useful property of $1\mbox{-}\mathrm{PLIN/lin}$ was presented in \cite[Lemma 4.7]{Yam10}.
A similar property can be shown in our setting of promise problem families by exploiting the connection between $1\mbox{-}\mathrm{PLIN/lin}$ and $\onep$.

\begin{lemma}\label{property-onep}
Let $\{(L_n^{(+)},L_n^{(-)})\}_{n\in\nat}$ be any family in $\onep$ over alphabet $\Sigma$. There exists a polynomial $p$ that satisfies the following statement. Let $n$, $m$, and $l$ denote any numbers in $\nat^{+}$ with $l\leq m-1$ and $p(n)<|\Sigma|^{m-l}$.
If $1\leq |(L_n^{(+)}\cup L_n^{(-)})\cap \Sigma^{m}|<|\Sigma|^{m}$, then there exist a number $k\in\nat^{+}$ and a set $S=\{w_1,w_2,\ldots,w_{k}\} \subseteq \Sigma^{m-l}$ with $k\leq p(n)$ for which the following implication holds: for any subset $R\subseteq \Sigma^l$, if $|\{a(y) \mid y\in R \}|\geq 2^{k}$, then it follows that, for any $x\in \Sigma^{m-l}$, there exists a pair $y,y'\in R$ satisfying that  $xy\in L_n^{(+)}$ and $xy'\notin L_n^{(+)}$, where $a(y) = a^{(y)}_1a^{(y)}_2\cdots a^{(y)}_{k}$, and  $a^{(y)}_i=1$ if $w_iy\in L_n^{(+)}$, and $a^{(y)}_i=0$ if $w_iy\notin L_n^{(+)}$ for any index $i\in[k]$.
\end{lemma}

Intuitively, the lemma says that it is possible to choose a certain number of prefixes $w_1,w_2,\ldots,w_k$ such that, for any prefix $x$, two suffixes $y$ and $y'$ can differentiate between $xy$ and $xy'$ regarding their memberships to $L^{(+)}_n$ if at least $2^k$ suffixes $y$ produce distinct sequences $a(y)$ expressing the memberships of $w_iy$'s to $L^{(+)}_n$.

\ms

Using Lemma \ref{property-onep}, we wish to prove the desired class separation of Theorem \ref{parity-vs-onep}.
For two binary strings $x$ and $y$ of equal length, the \emph{bitwise inner product} $x\odot y$ of $x$ and $y$ is defined to be $\sum_{i=1}^{n}x_iy_i$, where $x=x_1x_2\cdots x_n$ and $y=y_1y_2\cdots y_n$ with $x_i,y_i\in\{0,1\}$ for all $i\in[n]$.

\vs{-2}
\begin{proofof}{Theorem \ref{parity-vs-onep}}
Assuming that $\oneparity\subseteq \onep$, we intend to lead to a  contradiction. We first define an example family $\LL_{\parity} =\{(L_n^{(+)},L_n^{(-)})\}_{n\in\nat}$ of promise problems as follows. Fix $n\in\nat$ arbitrarily and define $J_n$ to be the set $\{u_1\dollar u_2\dollar \cdots \dollar u_n\mid \forall i\in[n] [u_i\in\{0,1\}^{\floors{\log{n}}}]\}$, where $\dollar$ is a designated separator not in $\{0,1\}$. Notice that $|u|=n(\floors{\log{n}}+1) -1$ for each string $u\in J_n$.
Given two strings $u,v\in J_n$, we further define $u\:\square\: v$ to express the value   $\sum_{j=1}^{n}(u_j\odot v_j\:(\mathrm{mod}\:2))$, where $u=u_1\dollar u_2\dollar \cdots \dollar u_n$ and  $v=v_1\dollar v_2\dollar \cdots \dollar v_n$.
We further define $L_n^{(+)} = \{u\# v\mid u,v\in J_n, u\:\square\: v\: \equiv 0 \:(\mathrm{mod}\:2)\}$ and $L_n^{(-)}= \{u\# v\mid u,v\in J_n\} - L_n^{(+)}$. It follows that $|u\# v|=2|u|+1 = 2n(\floors{\log{n}}+1)-1$ for any string $u\# v\in L_n^{(+)}\cup L_n^{(-)}$.

Let us consider the 1nfa $N_n$ that works as follows. On input $x$ of the form $u\# v$ with $u,v\in J_n$, guess  (i.e., nondeterministically choose) a number $i\in[n]$, read $u_i$ out of  $u$, and remember $(i,u_i)$ in the form of inner states.
This is possible because of $|u_i|=\floors{\log{n}}$.
After passing $\#$, read $v_i$ out of $v$ and calculate the value $a_i=u_i\odot v_i\;(\mathrm{mod}\:2)$. If $a_i=1$, then accept $x$, and otherwise, reject $x$.
From the definition of $N_n$, it follows that $\LL_{\parity}$ is in  $\oneparity$. Our assumption then concludes that $\LL_{\oplus}$ belongs to  $\onep$.

Apply Lemma \ref{property-onep} and take a polynomial $p$ provided by the lemma.  We fix a sufficiently large number $n\in\nat$ satisfying $p(n)<2^{(m+1)/2}$, where $m=2n(\floors{\log{n}}+1)-1$.
We then set $l=(m+1)/2$. Notice that $m-l=l-1$. Since $m>1$, we obtain $l<m-1$.
For these values $(l,m)$, there is a subset $S=\{w_1,w_2,\ldots,w_{k}\}$ of $\{0,1\}^{m-l}$ that satisfies the lemma. Note that $1\leq k\leq p(n)$. For each string $r=r_1r_2\cdots r_{k}$ with $r_i\in\{0,1\}$ for all $i\in[k]$,
we choose a string $y_r\in\{0,1\}^{l}$ satisfying $w_i\:\square\: y_r \equiv r_i\:(\mathrm{mod}\:2)$ for all indices $i\in[k]$.
Remember that all chosen $y_r$'s are distinct.
We further take a string $x\in\{0,1\}^{m-l}$ satisfying $x\:\square\: y_r \equiv 0\:(\mathrm{mod}\:2)$ for all $r\in\{0,1\}^{k}$.
To draw a conclusion of the lemma, we take the set $R = \{\#y_r\mid r\in \{0,1\}^{k}\}$.
For each string $\#y$ in $R$, we define $a(\#y) = a_1a_2\cdots a_{k}$ by setting $a_i \equiv w_i\:\square\: y\:(\mathrm{mod}\:2)$ for any index $i\in[k]$.
The definition of $R$ implies that $|\{a(\#y) \mid \#y\in R \}| \geq   2^{k}$.
By the conclusion of the lemma, a certain pair $\#y,\#y'\in R$ exists and it satisfies both $x\#y\in L_n^{(+)}$ and $x\#y'\notin L_n^{(+)}$. In other words, $x\:\square\: y\equiv 1\:(\mathrm{mod}\:2)$ and $x\:\square y'\equiv 0\:(\mathrm{mod}\:2)$. This is in  contradiction to the choice of $x$. Therefore, we conclude that $\LL_{\oplus}\notin \onep$.
\end{proofof}


The proof of Lemma \ref{property-onep} requires a nice property regarding the success probability of 1pfa's. Recall the description of these 1pfa's from Section \ref{sec:class-onecequal}.
It is proven in \cite[Claim 2]{Yam22a} that, for any given 1pfa $M$, there exists another 1pfa $N$ with $L(M)=L(N)$ for which, for any $x$, if $M$ accepts $x$ with probability more than $1/2$ (resp., rejects $x$ with probability at least $1/2$), then $N$ accepts (resp., rejects) $x$ with probability \emph{more than} $1/2$.
It is possible to prove the same property for nonuniform families of unbounded-error 1pfa's.

\begin{lemma}\label{error-bound-qpfa}
Let $\LL=\{(L_n^{(+)},L_n^{(-)})\}_{n\in\nat} \in\onep$. There exist a polynomial $p$ and a family $\MM=\{M_n\}_{n\in\nat}$ of 1pfa's such that, for any $n\in\nat$, (1) $st(M_n)\leq p(n)$, (2) for all $x\in L_n^{(+)}$, $p_{acc,n}(x)>1/2$, and (3) for all $x\in L_n^{(-)}$, $p_{rej,n}(x)>1/2$.
\end{lemma}

We return to Lemma \ref{property-onep} and hereafter provide its proof.

\begin{proofof}{Lemma \ref{property-onep}}
Let $\LL=\{(L_n^{(+)},L_n^{(-)})\}_{n\in\nat}$ denote any family of promise problems over alphabet, say, $\Sigma$ in $\onep$.
By Lemma \ref{error-bound-qpfa}, we can take a family $\{N_n\}_{n\in\nat}$ of polynomial-size 1pfa's satisfying the following conditions for all indices $n\in\nat$:  (*) $p_{acc,n}(x)>1/2$ holds for any $x\in L_n^{(+)}$ and  $p_{rej,n}(x)>1/2$ holds for any $x\in L_n^{(-)}$.
In the following argument, we only consider valid strings in $L_n^{(+)}\cup L_n^{(-)}$ for any $n\in\nat$.
Let $Q_n$, $Q_{acc,n}$, and $Q_{rej,n}$ respectively denote the sets of all inner states, of all accepting inner states, and of all rejecting inner states of $N_n$. Note that there is a polynomial $p$ satisfying $|Q_n|\leq p(n)$ for all $n\in\nat$.

We fix $n$ arbitrarily. Let us consider three vectors $\nu_{ini}^{(n)}$,  $\xi_{Q_{acc,n}}$, and $\xi_{Q_{rej,n}}$, and a set $\{M_{\sigma}^{(n)}\}_{\sigma\in\check{\Sigma}}$ of probability transition matrices of $N_n$, where $\check{\Sigma}=\Sigma\cup\{\rhd,\lhd\}$.
Fix two numbers $m,l\in\nat^{+}$ satisfying $l\leq m-1$ and $p(n)<|\Sigma|^{m-l}$. Assume that $1\leq |L^{(n)}\cap\Sigma^{m}|<|\Sigma|^{m}$. We then take a maximal subset $S'$ of linearly independent vectors in $V_n= \{\nu_{ini}^{(n)} M_{\triangleright w}^{(n)} \mid w\in \Sigma^{m-l}\}$. Since $V_n$ contains only $|Q_n|$-dimensional vectors, there are at most $p(n)$ linearly independent vectors. Thus, $|S'|\leq p(n)$ follows.
Next, we define $S=\{w\in\Sigma^{m-l}\mid \nu_{ini}^{(n)} M_{\triangleright w}^{(n)} \in S'\}$. By setting $k=|S|$, we obtain $|S'| \leq k\leq p(n)$.
Let us assume that $S$ is expressed as $\{w_1,w_2,\ldots,w_{k}\}$.
For each string $y\in\Sigma^l$, we write $a(y)$ for the binary string $a_1a_2\cdots a_{k}$ satisfying that, for any index $i\in[k]$, $a_i=1$ if $w_iy\in L_n^{(+)}$, and $a_i=0$ if $w_iy\notin L_n^{(+)}$.
We then choose an arbitrary subset $R$ of $\Sigma^l$ for which $|\{a(y) \mid  y\in R  \}|\geq 2^{k}$. In the case that no such $R$ exists,
the lemma is vacuously true.

Fix an arbitrary string $x\in \Sigma^{m-l}$ and take a series $\{\alpha_i\}_{i\in[k]}$ of real numbers satisfying $\nu_{ini}^{(n)} M_{\triangleright x}^{(n)} = \sum_{i\in[k]}\alpha_i(\nu_{ini}^{(n)} M_{\triangleright w_i}^{(n)})$.
We next claim that $\sum_{i\in[k]}\alpha_i=1$. Given a (row) vector $\xi$ of $|Q_n|$-dimension, we succinctly write $\|\xi\|_{sum}$ for the sum of the values of all entries of $\xi$.
Note that $\|\nu_{ini}^{(n)} M_{\triangleright  w_i}^{(n)}\|_{sum}=1$ holds for all $i\in[k]$ since $M^{(n)}_{\triangleright w_i}$ is a stochastic matrix and $\|\nu_{ini}^{(n)}\|_{sum}=1$.
It then follows that  $\|\sum_{i\in[k]}\alpha_i(\nu_{ini}^{(n)} M_{\rhd w_i}^{(n)})\|_{sum} = \sum_{i\in[k]}\alpha_i \| \nu_{ini}^{(n)} M_{\triangleright w_i}^{(n)}\|_{sum}= \sum_{i\in[k]}\alpha_i$. From this fact, we instantly obtain $\sum_{i\in[k]}\alpha_i=1$ because of $\|\nu_{ini}^{(n)} M_{\triangleright x}^{(n)}\|_{sum} = 1$.

Given any $y\in\Sigma^l$ and any $i\in[k]$, there exists a series  $\{\beta_{iy}\}_{i,y}$ of real numbers for which $p_{acc,n}(w_iy) = 1/2+\beta_{iy}$. Note that $-1/2\leq \beta_{iy} \leq 1/2$ but $\beta_{iy}\neq0$ due to Condition (*).
It then follows that $p_{acc,n}(xy)= \nu_{ini}^{(n)} M_{\triangleright x}^{(n)} M_{y\triangleleft}^{(n)} \xi_{Q_{acc,n}}^T =
\sum_{i\in[k]} \alpha_i (\nu_{ini}^{(n)} M_{\triangleright w_i}^{(n)} ) M_{y\triangleleft}^{(n)} \xi_{Q_{acc,n}}^T
= \sum_{i\in[k]}\alpha_i p_{acc,n}(w_iy) =\sum_{i\in[k]} \alpha_i (\frac{1}{2}+\beta_{iy}) = \frac{1}{2}\sum_{i\in[k]}\alpha_i + \sum_{i\in[k]}\alpha_i\beta_{iy}
= \frac{1}{2} + \sum_{i\in[k]} \alpha_i\beta_{iy} = \frac{1}{2}+ \sum_{i\in[k]} (-1)^{1-a(y)[i]}|\beta_{iy}|\alpha_i$, where $a(y)[i]$ denotes the $i$th bit of the string $a(y)$.

The desired strings $y$ and $y'$ are defined as follows. We take a $k$-bit-long string $b=b_1b_2\cdots b_{k}$ such that $b_i=1$ if $\alpha_i>0$, and $b_i=0$ if $\alpha_i\leq 0$ for any index $i\in[k]$.
From this $b$, a string $y\in R$ is chosen to satisfy $a(y)=b$.
Such a string $y$ exists because $|\{a(y')\mid y'\in R\}|\geq 2^{k}$ implies $b\in\{a(y')\mid y'\in R\}$. For this $y$, we obtain
$(-1)^{1-a(y)[i]}\alpha_i = (-1)^{1-b_i}\alpha_i =|\alpha_i|$. Therefore, it follows that
$\sum_{i\in[k]} (-1)^{1-a(y)[i]} |\beta_{iy}|\alpha_i = \sum_{i\in[k]} |\beta_{iy}||\alpha_i|>0$. This instantly implies $p_{acc,n}(xy)>1/2$. For the other string $y'$, we first define $\bar{b}$ to be the bitwise negation of $b$ and then choose $y'$ to satisfy  $a(y')=\bar{b}$. This implies that $(-1)^{1-a(y')[i]}\alpha_i = (-1)^{1-\overline{b_i}}\alpha_i = (-1)^{b_i}\alpha_i = -|\alpha_i|$. It thus follows that $\sum_{i\in[k]}(-1)^{1-a(y')[i]}|\beta_{iy'}|\alpha_i = - \sum_{i\in[k]} |\beta_{iy'}||\alpha_i|<0$. From this inequality,  $p_{acc,n}(xy')<1/2$ follows. We thus conclude that $xy\in L_n^{(+)}$ and $xy'\notin L_n^{(+)}$.
\end{proofof}

\subsection{Implications to Partial Function Classes}\label{sec:implications}

In Sections \ref{sec:oneu-class}--\ref{sec:onep-class}, we have demonstrated numerous separations among counting complexity classes.
Here, we shift our attention to separations of function classes and wish to  strengthen Lemma \ref{inclusion-1F-to-1sharp} by verifying the two additional separations shown below.
The proofs of these separations heavily relay on the results of Theorem \ref{N-vs-cequal}(1) and $\oned\neq \onen$.

\begin{proposition}\label{separation-onef}
(1) $\onesharp\neq\onegap_{\geq0}$. (2) $\onef_{\integer}\neq \onegap$.
\end{proposition}

\begin{proof}
(1) The following argument requires in part the use of Theorem \ref{N-vs-cequal}(1). To show the desired separation of $\onesharp\neq \onegap_{\geq0}$, it suffices to prove that $\onesharp=\onegap_{\geq0}$ implies $\onecequal=\co\onen$ because Theorem \ref{N-vs-cequal}(1) implies $\onecequal\neq \co\onen$, contradicting $\onecequal=\co\onen$. Assuming $\onesharp=\onegap_{\geq0}$, let us consider a family $\LL=\{(L_n^{(+)},L_n^{(-)})\}_{n\in\nat}$ in $\onecequal$ and take a family $\GG=\{(g_n,D_n)\}_{n\in\nat}$ in $\onegap$ such that, for all $n\in\nat$, $g_n(x)=0$ holds for all $x\in L_n^{(+)}$ and  $g_n(x)\neq0$ holds for all $x\in L_n^{(-)}$.
Notice that $D_n$ equals $L^{(+)}_n\cup L^{(-)}_n$ for all $n\in\nat$.
We then set $f_n(x)=g_n(x)^2$ for all $x\in D_n$.
It then follows that $f_n(x)=0$ for all $x\in L_n^{(+)}$, and  $f_n(x)>0$ for all $x\in L_n^{(-)}$.
By the closure property of $\onegap$ under multiplication (Lemma \ref{closure-onesharp}),
the family $\FF = \{(f_n,D_n)\}_{n\in\nat}$ belongs to  $\onegap_{\geq0}$. Our assumption further implies that $\FF$ falls in $\onesharp$. Consequently, $\LL$ is in $\co\onen$.

(2) We next prove that $\onef_{\integer}\neq \onegap$.
Our goal is to show that (*) $\onef_{\integer}=\onegap$ implies $\oned=\onep$. Since $\onen\subseteq \onep$, $\oned\neq\onen$ \cite{SS78} implies $\oned\neq\onep$. Therefore, if (*) is true, then we can conclude that $\onef_{\integer}\neq\onegap$. For the proof of (*), we now assume that $\onef_{\integer} = \onegap$. Let $\LL$ denote any family in $\onep$ and take a family $\FF = \{(f_n,D_n)\}_{n\in\nat}$ in $\onegap$ that ``witnesses'' the membership of $\LL$ to $\onep$; namely, $f_n(x)>0$ for all $x\in L_n^{(+)}$ and $f_n(x)\leq 0$ for all $x\in L_n^{(-)}$.
By our assumption, $\FF$ belongs to $\onef_{\integer}$, in other words, $\{(f_n^{(trans)},D_n)\}_{n\in\nat}$ is in $\onef$.
There exists a family $\{M_n\}_{n\in\nat}$ of polynomial-size 1dft's $M_n$  computing $f_n^{(trans)}$ on $D_n$ for all $n\in\nat$.
Let us consider the following 1dfa $N_n$. On input $x\in D_n$, run $M_n$ on $x$ and accept it exactly when $M_n$ produces a string of the form $1y$. It then follows that $f_n^{(trans)}(x)>0$ iff $trans(f_n(x))\in 1\Sigma^*$ iff $N_n$ accepts $x$.
Thus, the family $\{N_n\}_{n\in\nat}$ indeed solves $\LL$. This implies that $\LL$ is in $\oned$, as requested.
\end{proof}

\section{Relations to Nonuniform Families of Pushdown Automata}\label{sec:relation-pushdown}

In this section, we turn our attention to pushdown automata families, which were initially studied in \cite{Yam21a}.
It is proven in \cite{Yam21a} that $\onen$ and $\onedpd$ are incomparable; namely, both $\onen\nsubseteq \onedpd$ and $\onedpd\nsubseteq \onen$ hold. In what follows, we intend to strengthen the former separation to $\oneu\nsubseteq \onedpd$ and the latter one to $\onedpd\nsubseteq \onep$.
Since $\onen\subseteq \onenpd$, $\onedpd\neq\onenpd$ also follows from $\onen\nsubseteq \onedpd$.
This shows a strength and a limitation of counting complexity classes.

\begin{theorem}\label{onedpd-separation}
(1) $\oneu\nsubseteq \onedpd$.
(2) $\onedpd \nsubseteq \onep$
\end{theorem}

To understand the behaviors of a 1dpda, the key is the stack height history of the 1dpda. An analytic technique based on Kolmogorov complexity was used in \cite{Yam22e} for stack height history to prove an analogue of the so-called \emph{swapping lemma} \cite{Yam08,Yam16}.

A core argument in the following proof of (1) roughly goes as follows. Assume that an input $w$ of high Kolmogorov complexity is partitioned into $xyz$ so that a given 1dpda stores the information $a\gamma$ obtained from $x$ onto a stack, accesses the string $a$ while processing $y$, and consumes the entire string $\gamma$ for $z$. If $y$ is relatively long and uniquely determined (from little information), then it is possible to significantly compress $w$, a contradiction to the choice of $w$.

\begin{proofof}{Theorem \ref{onedpd-separation}}
(1) Since $\onedpd=\co\onedpd$ \cite{Yam21a}, $\oneu\subseteq\onedpd$ is equivalent to $\co\oneu\subseteq\onedpd$.
In what follows, we intend to prove that $\co\oneu\nsubseteq \onedpd$.
Let us define a family $\LL = \{(L_n^{(+)},L_n^{(-)})\}_{n\in\nat}$ of promise problems, which will be used for the desired separation, as follows.
Let $\Sigma=\{0,1\}$ and take a designated separator $\#$ not in $\Sigma$. For simplicity, we write $\Sigma_{\#}$ for $\Sigma\cup\{\#\}$ and $\Theta_{\#}$ for $\Sigma^*\cup \Sigma^*\#\Sigma^*$. For each index $n\in\nat$, we intend to view each string $x$ in $\Sigma^{n^2}$ of the form $x_1x_2\cdots x_n$ with $x_i\in\Sigma^n$ for each $i\in[n]$ as a series $(x_1,x_2,\ldots,x_n)$ of separate $n$-bit strings. In this series, the notation $(x)_{(e)}$ denotes its $e$th entry $x_e$.
For clarity reason, hereafter, each $(x)_{(e)}$ is referred to as a \emph{block} of $x$. We define  $L_n^{(+)} = \{ u\# v\mid u,v\in \Sigma^{n^2}, \exists! e\in[n]((u)_{(e)}\neq (v)_{(e)})\}$ and $L_n^{(-)} = \{u\# v\mid u,v\in \Sigma^{n^2}\} - L_n^{(+)}$. It then follows that $\LL$ belongs to $\oneu$.

Since $\co\LL\in\co\oneu$, it suffices to verify that $\co\LL\notin \onedpd$. Toward a desired contradiction, we assume that $\co\LL\in \onedpd$ and
take a nonuniform family $\{M_n\}_{n\in\nat}$ of polynomial-size 1dpda's that solves $\co\LL$; namely, all strings in $L_n^{(-)}$ are accepted by $M_n$ and all strings in $L_n^{(+)}$ are rejected by $M_n$.
For each index $n\in\nat$, let $M_n=(Q_n,\Sigma_{\#}, \{\rhd,\lhd\}, \Gamma_n,\delta_n, q_{0,n}, \bot, Q_{acc,n},Q_{rej,n})$ and let $e_n$ denote the push size of $M_n$.
There is a polynomial $p$ satisfying $ssc(M_n) =|Q_n||\Gamma_n^{\leq e_n}|\leq p(n)$ for all $n\in\nat$.
It is possible to assume that $M_n$ has only one accepting state, say, $q_{acc,n}$. Moreover, we assume that $M_n$ always empties its stack at the end of computation.
This is always possible by making a series of $\lambda$-moves after reading $\lhd$. See also \cite{Yam21b} for relevant issues.
A \emph{configuration} of $M_n$ is a triplet $(q,\sigma w,\gamma)$ indicating that $M_n$ is in inner state $q$, a tape head is scanning $\sigma$, $\sigma w$ is a suffix of an input together with the endmarker(s), and $\gamma$ is the current stack content. The notation $\vdash_{M_n}$ refers to a transition between two configurations in a single step and $\vdash^*_{M_n}$ is the  transitive closure of $\vdash_{M_n}$.

In the following argument, we fix a sufficiently large number $n\in\nat$ and introduce several notations to simplify the argument.
Let $T_n = Q_n\times Q_n\times \Gamma_n^{(-)}$, where $\Gamma_n^{(-)} = \Gamma_n-\{\bot\}$.
For each element $(q_1,q_2,a)$ in $T_n$,  $D_{q_1,q_2,a}$ denotes the set of all strings $y\in\Theta_{\#}$ such that, for any $z\in\Theta_{\#}$ and any $\gamma\in(\Gamma^{(-)})^*\bot$, $(q_1,yz,a\gamma)\vdash^*_{M_n} (q_2,z,\gamma)$ holds and, while reading $y$, the stack height of $M_n$ does not go below $|\gamma|$; namely, $M_n$ does not access any symbol in $\gamma$. Given a string $u$, we denote by $\gamma_{u}$ the stack content obtained just after reading off $\rhd u$.

We write $I_n$ for the set $\{u\# v\mid u,v\in \Sigma^{n^2}\}$.
Let us recall the Kolmogorov complexity $C(x)$ of string $x$ from the proof of Theorem \ref{separate-co-oneu}.
By a simple counting argument (as in the proof of Theorem \ref{separate-co-oneu}), we can show that there exists a string $u\in \Sigma^{n^2}$ for which $C(u)$ is at least $n^2-1$.
Assume that, on the input $u\# u$, $M_n$ produces a deterministic computation in which $M_n$ enters a unique inner state $q$ with a stack content $\gamma_{u\#}$ obtained just after reading $\#$.
For any element $(q_1,q_2,a)\in T_n$, $E_{q_1,q_2,a}$ denotes the set of all strings $y$ such that, for  certain elements $x,z,a,\gamma$ satisfying $u\# u = xyz$,  $(q_0,\rhd xyz\lhd,\bot)\vdash^*_{M_n} (q_1,yz\lhd,a\gamma)\vdash^*_{M_n} (q_2,z\lhd,\gamma) \vdash^*_{M_n} (q_{acc,n},\varepsilon,\bot)$ holds.
For such a decomposition $(x,y,z)$ of $u\# u$, the notation $sh_{u}(x,y)$ expresses the stack height $|\gamma|$.
Clearly, $E_{q_1,q_2,a}\subseteq D_{q_1,q_2,a}$ follows. We set  $\tilde{E}_n$ to be the union $\bigcup_{(q_1,q_2,a)\in T_n} E_{q_1,q_2,a}$.

For notational simplicity, we hereafter write $\gamma_0$ for $\gamma_{u\#}$, which is the stack content obtained just after reading $\#$.
We first prove that the stack height $|\gamma_0|$ is at least $\alpha_n = \frac{n^2}{\log{p(n)}}-3$.
To show this statement, we assume that $|\gamma_0|<\alpha_n$. Since $q\in Q_n$ and $\gamma_0\in (\Gamma_n)^{|\gamma_0|}$,  it follows that  $C(q)\leq \ceilings{\log|Q_n|} +O(1) \leq \log{p(n)} +O(1)$ and $C(\gamma_0) \leq |\gamma_0|\cdot \ceilings{\log|\Gamma_n|} +O(1) \leq \alpha_n\log{p(n)} +O(1)$.
The last inequality comes from the fact of $|Q_n||\Gamma_n^{\leq e_n}|\leq p(n)$. It is possible to uniquely identify $u$ by cycling through all strings $v\in\Sigma^{n^2}$ and running $M_n$ on $v$ starting with  $(q,\gamma_0)$.
We thus obtain  $C(u)\leq C(n) + C(\pair{q,\gamma_0}) \leq \log{n} + O(1) \leq \alpha_n\log{p(n)}+ \log{p(n)}+\log{n}+ O(1)$. Because  $C(u)\geq n^2-1$ and $n$ is sufficiently large, we conclude that  $\alpha_n\log{p(n)} +\log{p(n)}+\log{n}+ O(1)\geq n^2$. This inequality gives a lower bound of $\alpha_n$ as $\alpha_n\geq \frac{n^2}{\log{p(n)}} - \frac{\log{p(n)}+\log{n}+O(1)}{\log{p(n)}}\geq \frac{n^2}{\log{p(n)}}-2$, which contradicts the definition of $\alpha_n$.

We next prove that, for each element $(q_1,q_2,a)\in T_n$, if $E_{q_1,q_2,a}\cap \Sigma_{\#}^{\leq n^2+1}\neq\setempty$, then $D_{q_1,q_2,a}$ has at most one element. Toward a contradiction, we assume that there are two distinct strings $y$ and $y'$ in $D_{q_1,q_2,a}$ such that either $y$ or $y'$ (or both) is in $E_{q_1,q_2,a}\cap \Sigma_{\#}^{\leq n^2+1}$.
Here, we consider the case where $y$ is in $E_{q_1,q_2,a}$ satisfying $u\#u = xyz$ for two appropriate strings $x$ and $z$.
Since $y,y'\in D_{q_1,q_2,a}$, we can freely replace this $y$ in $u\#u$ by $y'$. The obtained string $xy'z$ must be accepted as well. However, since $y\neq y'$ and $|y|\leq n^2+1$, $xy'z$ does not have the form $v\# v$ for any $v\in \Sigma^{n^2}$. This is a  clear contradiction.

Recall that the stack height $|\gamma_0|$ is at least $\alpha_n$.  We wish to show that there exists a string $y$ satisfying the following  condition:
\begin{center}
(*) $y\in \tilde{E}_n$ and $\alpha_n\leq |y|\leq n^2+1$.
\end{center}
Assuming that no string $y$ satisfies Condition (*), we take six strings $x,z,y_1,y_2,y_3,y_4$ with $r\#r = xyz$ and $y=y_1y_2\#y_3y_4$ for which $y$ and $y_2\#y_3$ are in $\tilde{E}_n$, $|y|>n^2+1$, $|y_2\#y_3|<\alpha_n$, $sh_u(x,y) = sh_u(xy_1,y_2\#y_3)+1$,  $M_n$'s stack height does not go below $|\gamma_x|-1$ while reading $y_2$, and $M_n$'s stack height does not go below $|\gamma_{xy}|-1$ while reading $y_4$.
Note that there are two elements $(q'_1, q'_2, a')$ and $(q''_1,q''_2,a'')$ in $T_n$ for which $y\in E_{q'_1,q'_2,a'}$ and $y_2\#y_3\in E_{q''_1,q''_2,a''}$.
Let us study the following two cases (i)--(ii) separately.

(i) If $|y_1y_2|\geq|y_3y_4|$, then the uniqueness of $y_1$ helps us  construct $y_1$ from the information on $(n,|y_1|,q'_1,q''_1,a',a'',b)$, where $b$ denotes a unique stack symbol satisfying $\gamma_{xy_1} = b\gamma_{x}$. Since $u=xy_1y_2$, it is therefore possible to construct $u$ algorithmically from $(n,x,|y_1|,y_2,q'_1,q''_1,a',a'',b)$. We then remark that $C(|y_1|)\leq \log|y_1|+O(1)\leq 3\log{n}$ and $|y_1|=|y|-|y_2\#y_3|-|y_4|\geq n^2+1-\alpha_n- \frac{n^2+1}{2} \geq \frac{n^2+1}{2}-\alpha_n$
since $|y_4|\leq \frac{|y|}{2}=\frac{n^2+1}{2}$.
Thus, we obtain $C(\pair{x,y_2}) \leq |x| + |y_2| + 2\log{n} +O(1) \leq n^2-|y_1| +2\log{n}+O(1) \leq \frac{n^2}{2} +\alpha_n +3\log{n}$.
It thus follows that $C(u) \leq C(n)+ C(|y_1|)+C(\pair{x,y_2})+C(\pair{q'_1,q''_1})+C(\pair{a',a''})+C(b) +3\log{n} +O(1)  \leq \frac{n^2}{2} + \alpha_n + 4\log{n}+4\log{p(n)}+O(1) \leq \frac{2n^2}{3}$.
This contradicts the inequality $C(u)\geq n^2-1$ since $n$ is sufficiently large.

(ii) In contrast, if $|y_3y_4|>|y_1y_2|$, then we can construct $y_4$ from $(n,|y_4|,q'_2,q''_2,b)$, where $b$ satisfies $\gamma_{u\#y_3} = b\gamma_{u\#y_3y_4}$. Since $u=y_3y_4z$, $u$ is algorithmically constructible from $(n,|y_4|,y_3,z,q'_2,q''_2,b)$. An argument similar to the first case leads to a contradiction. Therefore, Condition (*) holds.

To wrap up our argument, we take a string $y$ that satisfies Condition (*).
For this string $y$, there is an element $(q_1,q_2,a)$ in $T_n$ for which $D_{q_1,q_2,a}$ contains $y$.
There are two cases to discuss separately.
In the first case where $y$ does not include $\#$, we take $z=z_1\# z_2$ so that $u=xyz_1=z_2$.
Since $D_{q_1,q_2,a}=\{y\}$, $y$ is uniquely identified from the information on $(n,q_1,q_2,a)$.
This implies that $u$ is algorithmically constructible from $(n,x,q_0,q_1,q_2,a,z_1)$. It follows that $C(u)\leq C(n)+ C(\pair{q_0,q_1,q_2})+C(a)+ C(\pair{x,z_1})+O(1)\leq n^2-\alpha_n + 3\log{n} + 3\log{p(n)}+O(1) \leq n^2 - \frac{\alpha_n}{2}$. This contradicts the fact that $C(u)\geq n^2-1$ because $\alpha_n$ is strictly increasing as $n$ grows. 
Next, we examine the second case where $y$ contains $\#$. Let $y=y_1\#y_2$ and  $x=y_2x_2$ for appropriate strings $x_2,y_1,y_2$. Since $u$ equals $y_2x_2y_1$, we apply an argument similar to the first case  using the extra information on $x_2$ and the length $|y_2|$ of $y_2$ and then derive a contradiction.


(2) Let us recall the set $J_n$ from the proof of Theorem  \ref{parity-vs-onep}. We further define $L_n^{(+)} = \{u^R\# v\mid u,v\in J_n, \sum_{i=1}^{n}u_i\odot v_i\equiv 1\:(\mathrm{mod}\:2)\}$ and $L_n^{(-)} = \{u^R\# v\mid u,v\in J_n\} - L_n^{(+)}$, where $u^R$ is the \emph{reverse} of $u$. It is easy to solve the promise problem $(L_n^{(+)},L_n^{(-)})$ by a 1dpda defined as follows. On input $u\# v^R$ with $u,v\in J_n$, push $u$ into a stack. After reading $\#$, pop $n$ strings $(u^R_1,u^R_2,\ldots,u^R_n)$ in this order and calculate all values $a_i=u_i\odot v_i\:(\mathrm{mod} \:2)$. Finally, compute $a=\sum_{i=1}^{n}a_i\:(\mathrm{mod}\:2)$ and enter an accepting state or a rejecting state depending respectively on $a=1$ or $a=0$. Since this 1dpda uses only polynomially many inner states in $n$, the family $\LL_{\odot}=\{(L_n^{(+)},L_n^{(-)})\}_{n\in\nat}$ belongs to $\onedpd$.

For the claim of $\LL_{\odot}\notin\onep$, its proof can be carried out  analogously to the proof of Theorem \ref{parity-vs-onep}.
\end{proofof}


Since $\LL_{\odot}$ in the above proof can be solved by an appropriate 1-turn\footnote{A \emph{turn} of a pushdown automaton is a transition change between an increasing phase and a non-increasing phase of stack height of the pushdown automaton. See, e.g., \cite{GS66,Yam22d}. This notion should not be confused with a ``turn'' of a tape head of a finite automaton.} 1dpda family, we can strengthen Theorem \ref{onedpd-separation}(2) to $\mathrm{1t1DPD}\nsubseteq \onep$, where the prefix ``1t'' indicates ``1 turn''.

\begin{corollary}
$\mathrm{1t1DPD}\nsubseteq \onep$.
\end{corollary}


We close this section by demonstrating the separation between $\onen$ and $\onenpd$.

\begin{proposition}\label{onen-vs-onenpd}
$\onen \subsetneqq \onenpd$.
\end{proposition}

\begin{proof}
The inclusion $\onen\subseteq \onenpd$ is trivial. We next assert that $\onen\neq \onenpd$. If $\onen=\onenpd$, then we obtain $\onedpd\subseteq \onen \subseteq \onep$. However, this inclusion contradicts the fact that $\onedpd\nsubseteq \onep$ of  Theorem \ref{onedpd-separation}(2). Therefore, $\onenpd$ differs from $\onen$, as requested.
\end{proof}

\section{Non-Closure Properties of 1\# under Functional Operations}\label{sec:functional-operation}

Under the typical functional operations of \emph{addition} and \emph{multiplication}, we have shown in Lemma \ref{closure-onesharp} the closure properties of $\onesharp$ and $\onegap$. Here, we wish to expand the scope of such functional operations. In this section, we particularly examine six additional functional operations.


Since standard division and subtraction are not applicable to natural numbers, we introduce their restricted variants for our purpose.
We define the \emph{integer division} $\oslash$ by setting $a\oslash b = \floors{a/b}$ when $b\neq0$ for any numbers $a,b\in\nat$.
The special case of integer division, $a\oslash 2$, is referred to as the \emph{integer division by two}. We also define the \emph{proper subtraction} $\ominus$ as $a\ominus b = a - b$ if $a\geq b$ and $a\ominus b=0$ otherwise for any nonnegative integers $a$ and $b$. The \emph{proper decrement} is a special case of proper subtraction, defined by $a\ominus 1$.


Now, we wish to expand the above number operations to ``functional operations'' in a natural way.
Given two partial functions $(f,D)$ and $(g,E)$ with $D,E\subseteq \Sigma^*$ for a fixed alphabet $\Sigma$, we define
$f\oslash g$, $f\oslash2$, $f\ominus g$, and $f\ominus 1$ by setting  $(f\oslash g)(x) = f(x)\oslash g(x)$, $(f\oslash2)(x)=f(x)\oslash2$, $(f\ominus g)(x) = f(x)\ominus g(x)$, and $(f\ominus 1)(x) = f(x)\ominus 1$ for all strings $x$ in $D\cap E$.
Furthermore, we introduce two more functional operations, $\max(f,g)$ and $\min(f,g)$, by setting $\max(f,g)(x)=\max\{f(x),g(x)\}$ and $\min(f,g)(x)=\min\{f(x),g(x)\}$ for any $x\in D\cap E$.

Hereafter, we intend to verify the non-closure properties of $\onesharp$ under the above-mentioned six functional operations.

\begin{theorem}\label{functional-operation}
The counting function class $\onesharp$ is not closed under the following six functional operations: minimum, maximum, proper subtraction, integer division, proper decrement, and integer division by two.
\end{theorem}

\begin{proof}
The non-closure properties under proper subtraction and integer division follow immediately from those under proper decrement and integer division by two. Therefore, it suffices to target the following four operations: (1) proper decrement, (2) integer division by two, (3) maximum, and (4) minimum. Following the arguments given in \cite{OH93}, we relate those non-closure properties to the separations of nonuniform state complexity classes shown in Section \ref{sec:relationship} so that the separation results of the section immediately lead to the desired non-closure properties of the theorem.

(1) We begin with the case of proper decrement. Toward a contradiction, assuming that $\onesharp$ is closed under proper decrement,  we wish to prove that $\onen\subseteq \onesp$. If this statement is true, then we obtain $\onen\subseteq \onecequal$ since $\onesp\subseteq \onecequal\cap\co\onecequal$ (Proposition \ref{onesp-vs-onecequal}).
However, this contradicts Theorem \ref{N-vs-cequal}(2) and we thus conclude that $\onesharp$ is not closed under proper decrement.

In what follows, we concentrate on proving that $\onen\subseteq \onesp$.
Let $\LL=\{(L_n^{(+)},L_n^{(-)})\}_{n\in\nat}$ denote any family of promise problems in $\onen$ and take a family $\{M_n\}_{n\in\nat}$ of polynomial-size 1nfa's that solves $\LL$.
For each index $n\in\nat$, we set $D_n$ to be $L_n^{(+)}\cup L_n^{(-)}$ and define $f_n(x)=\#M_n(x)$ for all strings $x\in D_n$.
Let us consider the function family  $\FF=\{(f_n,D_n)\}_{n\in\nat}$. We further define $h_n(x)=f_n(x)\ominus 1$ and set $\HH=\{(h_n,D_n)\}_{n\in\nat}$. Since $\FF\in\onesharp$, our assumption implies that $\HH\in\onesharp$. There is a family $\NN=\{N_n\}_{n\in\nat}$ of polynomial-size 1nfa's that ``witnesses'' the membership of $\HH$ to $\onesharp$.
By Lemma \ref{branching-normal-form}, we can assume that $N_n$ is in a branching normal form and it thus produces exactly $c^{|\rhd x\lhd|}$ ($= c^{|x|+2}$) computation paths for all valid inputs $x$, where $c>0$ is a certain fixed constant.
For convenience, we write $g_n(x)$ for $-h_n(x)$ for all strings $x\in D_n$.  Since   $\{(g_n,D_n)\}_{n\in\nat}$ belongs to $\onegap$ by Lemma \ref{one-gap-character}(1), we conclude from Lemma \ref{closure-onesharp} that the family $\FF+\GG = \{(f_n+g_n,D_n)\}_{n\in\nat}$ also belongs to $\onegap$.
We remark that, if $x\in L_n^{(+)}$ then $(f_n+g_n)(x) = f_n(x)+g_n(x)= \#M_n(x) - (\#M_n(x) - 1) =1$, and if $x\in L_n^{(-)}$ then $(f_n+g_n)(x) = f_n(x)+g_n(x)=0-0=0$. This implies that $\LL$ is in $\onesp$, as requested.

(2) We next focus on the integer division by two. Assume that $\onesharp$ is closed under integer division by two. From this assumption, we intend to derive the equality of $\onesp=\oneparity$. This clearly contradicts Proposition \ref{onesp-vs-oneparity}. Let $\LL=\{(L_n^{(+)},L_n^{(-)})\}_{n\in\nat}$ denote any family of promise problems in $\oneparity$.
There exists a family $\FF=\{(f_n,D_n)\}_{n\in\nat}$ of partial functions in $\onesharp$ that ``witnesses'' the membership of $\LL$ to $\oneparity$. Let us further take a family $\{M_n\}_{n\in\nat}$ of polynomial-size 1nfa's that ``witnesses'' the fact of ``$\FF\in\onesharp$''. It is possible to assume by Lemma \ref{branching-normal-form} that $M_n$ is in a branching normal form.

We define $\HH=\{(h_n,D_n)\}_{n\in\nat}$ by setting $h_n(x) = 2(f_n(x)\oslash 2)$ for all $n\in\nat$ and $x\in D_n$. Since $\HH$ is in $\onesharp$ by our assumption, there is another family $\{N_n\}_{n\in\nat}$ of polynomial-size 1nfa's that ``witnesses'' the membership of $\HH$ to $\onesharp$.
We then define $g_n(x)=c^{|x|+2}-f_n(x)$.
Note that, by Lemma \ref{one-gap-character}(1), the family $\GG = \{(g_n,D_n)\}_{n\in\nat}$ belongs to $\onegap$. Lemma \ref{closure-onesharp} then implies that $\FF+\GG = \{(g_n+h_n,D_n)\}_{n\in\nat}$ also falls in $\onegap$. Given any index $n$ and any string $x$, it follows that $x\in L_n^{(+)}$ implies $g_n(x)+h_n(x)=c^{|x|+2}$ and that $x\in L_n^{(-)}$ implies $g_n(x)+h_n(x)=c^{|x|+2}-1$. Hence, $\LL$ belongs to $\co\onesp$, which also equals $\onesp$ by Lemma \ref{parity-complement}. Therefore, we obtain  $\oneparity\subseteq \onesp$. Since $\onesp\subseteq \oneparity$ holds (by the first part of the proof of Proposition \ref{onesp-vs-oneparity}), we conclude that $\oneparity=\onesp$.

(3) We then target the maximum operation. In this case, we want to show that, if $\onesharp$ is closed under maximum, then $\onecequal$ collapses to $\onesp$. Since $\onecequal\neq \onesp$ by Proposition \ref{onesp-vs-onecequal}, we conclude the desired non-closure property of $\onesharp$.
We begin with taking
any family $\LL=\{(L_n^{(+)},L_n^{(-)})\}_{n\in\nat}$ of promise problems in $\onecequal$. There exists a family $\NN=\{N_n\}_{n\in\nat}$ of polynomial-size 1nfa's that ``witnesses'' the membership of $\LL$ to $\onecequal$.
By Lemmas \ref{branching-normal-form} and \ref{cequal-path-number}, we can assume that $\#N_n(x)=(2c)^{|x|+2}/2$ holds for any $x\in L_n^{(+)}$, and  $\#\overline{N}_n(x)<(2c)^{|x|+2}/2$ holds for any $x\in L_n^{(-)}$, where $c\geq1$ is a certain fixed constant.

Next, we define $D_n=L_n^{(+)}\cup L_n^{(-)}$ for all $n\in\nat$. Moreover, we set $f_n(x)=\#N_n(x)$ and $g_n(x)= (2c)^{|x|+2}/2 -1$ for all numbers $n\in\nat$ and all valid strings $x\in D_n$. Clearly, $\{(f_n,D_n)\}_{n\in\nat}$ and  $\{(g_n,D_n)\}_{n\in\nat}$ are in $\onesharp$.
From $f_n$ and $g_n$, we define $h_n=\max\{f_n,g_n\}$. Since $\onesharp$ is closed under maximum by our assumption, $\{(h_n,D_n)\}_{n\in\nat}$ also belongs to $\onesharp$.
It thus follows that $h_n(x)=(2c)^{|x|+2}/2$ for all $x\in L_n^{(+)}$, and  $h_n(x)=(2c)^{|x|+2}/2-1$ for all $x\in L_n^{(-)}$.
By setting $k_n(x)=-g_n(x)$, the family $\KK=\{(k_n,D_n)\}_{n\in\nat}$ belongs to $\onegap$. Lemma \ref{closure-onesharp} then implies that the family $\HH+\KK = \{(h_n+k_n,D_n)\}_{n\in\nat}$ also belongs to $\onegap$. Note that $x\in L_n^{(+)}$ implies  $h_n(x)+k_n(x)=1$ and that $x\in L_n^{(-)}$ implies $h_n(x)+k_n(x)=0$.
This indicates that $\LL$ is in $\onesp$.

(4) Finally, we consider the minimum operation. Assuming that $\onesharp$ is closed under minimum, we intend to prove $\onen=\oneu$ because this contradicts the early result of $\onen\neq\oneu$ \cite{Yam22b,Yam22c}. Let us consider an arbitrary family $\LL=\{(L_n^{(+)},L_n^{(-)}\}_{n\in\nat}$ in $\onen$ and take a family $\FF=\{(f_n,D_n)\}_{n\in\nat}$ in $\onesharp$ satisfying that $D_n=L_n^{(+)}\cup L_n^{(-)}$ for all $n\in\nat$ and that, for any $x\in D_n$, $x\in L_n^{(+)}$ implies $f_n(x)>0$, and $x\in L_n^{(-)}$ implies $f_n(x)=0$.
By setting $g_n(x)=1$ for all $n$ and all valid $x$, we define $h_n=\min\{f_n,g_n\}$. Note that the family $\{(h_n,D_n)\}_{n\in\nat}$ belongs to $\onesharp$ since  $\onesharp$ is closed under minimum. The definition of $h_n$ makes  $\LL$ fall in $\oneu$. We thus conclude that $\onen=\oneu$, as requested.
\end{proof}

\section{A Short Discussion and Open Problems}\label{sec:discussion}

Toward the full understandings of the essence of nondeterminism, counting has remained as an important research subject in computational complexity theory. To promote the better understandings of the nature of counting in various low complexity classes, we have initiated in this work a study of ``counting'' within the framework of nonuniform models of polynomial-size finite automata families.
This work has made a significant contribution to expanding and exploring the existing world of nonuniform (polynomial) state complexity classes \cite{BL77,Gef12,Kap09,Kap12,Kap14,KP15,SS78, Yam19a,Yam19b,Yam21a,Yam22a,Yam22b,Yam23b} with a special focus on ``counting''.
Throughout this work, we have demonstrated numerous containments and separations among counting complexity classes. These results are summarized in Figure \ref{fig:class-separations}.

For future research, hereafter, we briefly argue 7 interesting topics concerning ``counting''.

\renewcommand{\labelitemi}{$\circ$}
\begin{enumerate}
  \setlength{\topsep}{-2mm}%
  \setlength{\itemsep}{1mm}%
  \setlength{\parskip}{0cm}%

\item Unfortunately, Figure \ref{fig:class-separations} is not yet complete. For example, it is not yet clear that  $\onecequal\nsubseteq \oneparity$, $\oneparity\nsubseteq \onenpd$, and $\onecequal\nsubseteq \onedpd$ hold. Therefore, it is of importance to complete the figure by proving all the missing containments and separations.

\item It is known that the unbounded-error probabilistic complexity classes $\pp$ (polynomial time) and $\pl$ (logarithmic space) are both closed under union and intersection \cite{AO96,BRS95}. Is it true that $\onep$ is also closed under union and intersection? We conjecture that this is not the case but its proof is not yet known so far.

\item In Sections \ref{sec:basic-rel-among} and \ref{sec:functional-operation}, we have studied closure properties of $\onesharp$ and $\onegap$ under functional operations. Interestingly, many non-closure properties of $\onesharp$ are proven in Section \ref{sec:implications} and \ref{sec:functional-operation} by reducing them to separations of nonuniform state complexity classes. It is desirable to find more interesting functional operations under which $\onesharp$ and $\onegap$ are closed.

\item The number of ``turns'' in pushdown automata has played a key role in Section \ref{sec:relation-pushdown}, where we have briefly referred to the strength of the 1-turn restriction of $\onedpd$. It is possible to expand this notion to the \emph{$k$-turn restriction} of $\onedpd$, denoted $k\mathrm{t1DPD}$, for any fixed $k\in\nat^{+}$ and to discuss various structural properties of $k\mathrm{t1DPD}$ (as well as its nondeterministic variant, $k\mathrm{t1NPD}$). See, e.g.,
    \cite{GS66,Yam22d} for basic properties of $k$-turn 1dpda's.

\item In this work, we have not discussed ``two-way'' variants of nonuniform state complexity classes. Two-way head moves seem to dramatically change the landscape of nonuniform (polynomial) state complexity classes. Using two-way finite automata, we can define the two-way variants of $\onesharp$ and $\onegap$, which we analogously write $\twosharp$ and $\twogap$, respectively. These function classes help us define two-way variants of $\oneparity$, $\onecequal$, etc. To express them, we use the prefix ``2'', such as $\twoparity$ and $\twocequal$. For these complexity classes, is it true that, for instance,  $\twoparity\nsubseteq 2^{\oned}$ and $\twocequal\nsubseteq 2^{\oned}$? Here,  $2^{\oned}$ is defined similarly to $\oned$ but using families of ``exponential-size'' 1dfa's. See, e.g., \cite{Kap09,Kap12,Yam22a} for more information on this complexity class.

\item Let us recall from \cite{Kap14,SS78} that $\twod$ and $\twon$ are closely related to $\dl/\poly$ and $\nl/\poly$. When  all valid instances of the $n$th promise problem in a given family $\LL =\{(L^{(+)}_n,L^{(-)}_n)\}_{n\in\nat}$ are restricted to length at most polynomial in $n$ for a fixed polynomial,  we say that $\LL$ has a \emph{polynomial ceiling} \cite{Yam22a}. Given a complexity class $\CC$ of families of promise problems, the notation $\CC/\poly$ refers to a restriction of $\CC$ on families of polynomial ceilings. Using this succinct notation, for example, we obtain $\twod/\poly$ and $\twon/\poly$ \cite{Kap09,Kap12}. Refer also to \cite{Yam22a} for other complexity classes. In a similar fashion, we can introduce $\twosharp/\poly$ and $\twogap/\poly$ from $\twosharp$ and $\twogap$. What kind of  connection exists between $\twosharp/\poly$ (resp.,  $\twogap/\poly$) to $\sharpl/\poly$ (resp., $\gapl/\poly$)? Can we expand such a connection to log-space counting complexity classes (such as $\parityl$, $\cequall$, and $\pl$ \cite{AJ93,BJL+91,OH93})?

\item As we have briefly discussed in Section \ref{sec:def-counting-func}, counting and gap  functions can be computed by a restricted form of \emph{weighted automata}. We may be able to expand $\onesharp$ and $\onegap$ further using ``general'' weighted automata with positive rational or real weights. It is interesting to explore the properties of these generalized counting and gap functions.
\end{enumerate}

\let\oldbibliography\thebibliography
\renewcommand{\thebibliography}[1]{%
  \oldbibliography{#1}%
  \setlength{\itemsep}{-2pt}%
}
\bibliographystyle{alpha}

\end{document}